\documentclass[letterpaper,onecolumn]{IEEEtran}

\usepackage{graphicx} 

\usepackage{hyperref}

\usepackage[utf8]{inputenc} 
\usepackage[T1]{fontenc}
\usepackage{url}
\usepackage{ifthen}
\usepackage{cite}
\usepackage[cmex10]{amsmath} 
\usepackage{amssymb,dsfont,cancel,amsthm}
\usepackage{tikz}
\usetikzlibrary{plotmarks}
\usepackage{pgfplots}
\usetikzlibrary{calc}
\usetikzlibrary{shapes,arrows}
\usetikzlibrary{decorations.markings}
\usetikzlibrary{positioning}
\pgfplotsset{compat=1.10}
\usetikzlibrary{calc}
\usetikzlibrary{shapes,arrows}
\usetikzlibrary{decorations.markings}
\usepackage[draft]{tikzpeople}
\usetikzlibrary{math} 
\tikzmath{\x = 1.875; \y =1.25; \z=6.0; \b=.7;}

\DeclareMathOperator{\Tr}{Tr}

\newtheorem{theorem}{Theorem}

\newtheorem{definition}{Definition}

\newtheorem{lemma}{Lemma}
\newtheorem{corollary}{Corollary}

\newcommand{\Enc}{\mathsf{Enc}}
\newcommand{\Dec}{\mathsf{Dec}}
\newcommand{\Shp}{\mathsf{Shp}}

\title{One-Shot Secret Sharing with Monotone Access Structures over Classical-Quantum Broadcast Channels}
\author{Truman Welling, R\'emi A. Chou, and Aylin Yener
    \thanks{ T. Welling is with the Department of Electrical and Computer Engineering, The Ohio State University, Columbus, OH 43210. E-mail: welling.78@osu.edu. A. Yener is with the Departments of Electrical and Computer Engineering, Computer Science and Engineering, and Integrated Systems Engineering, The Ohio State University, Columbus, OH 43210. E-mail: yener@ece.osu.edu. R. Chou is with the Department of Computer Science and Engineering, The University of Texas at Arlington, Arlington, TX 76019. E-mail: remi.chou@uta.edu. This paper was presented in part at the 2025 IEEE International Symposium on Information Theory~\cite{welling-2025-SSQBC}.}
}
\date{April 2025}
\begin{document}
\maketitle
\begin{abstract}
   We consider a secret sharing setting with a monotone access structure involving a control node and $L$ users, connected via a classical–quantum broadcast channel whose input is controlled by the control node, referred to as the dealer. Unlike traditional secret sharing settings, where the dealer fully controls the shares given to each user, in our model, the dealer encodes the secret for transmission over the broadcast channel. This means that the shares received by users are perturbed by the channel and are not fully controlled by the dealer. Our main results are  achievable one-shot secret sharing rates, as well as converse bounds for arbitrary monotone access structures. 
   We further derive second-order and asymptotic achievable rates for arbitrary monotone access structures. In the special case where all shares are required to recover the secret, we show that our result coincides with the existing secret sharing capacity over classical channels.
\end{abstract}

\section{Introduction}
In a secret sharing protocol, a dealer distributes shares of a secret to participants such that authorized subsets of shares, defined by a predetermined access structure, can reconstruct the secret, while unauthorized subsets gain no information about the secret. The first secret sharing protocols, introduced independently in \cite{blakley-1979-SecretSharing} and  \cite{shamir-1979-SecretSharing}, are threshold schemes, where any subset of shares exceeding a certain size can recover the secret.
Secret sharing was extended to general access structures in \cite{ito-1989-GeneralAccess,benaloh-1988-GeneralAccess}. Their properties, including the characterization of optimal share size, have been studied extensively, e.g., \cite{beimel-2011-SecretSharing}.
Secret sharing has also been extended to the quantum domain, with schemes developed for sharing classical secrets \cite{hillery-1999-QSecretSharing,karlsson-1999-QSS} and quantum secrets \cite{cleve-1999-QTST,gottesman-2000-TheoryQSS,smith2000quantum}.

The secret sharing protocols mentioned up to this point all use the following two step structure: the dealer encodes the secret into shares then distributes them to participants. In both the classical and quantum settings, the dealer fully controls (knows) what each participant receives. Security in classical schemes relies
on the secrecy of the delivered shares, whereas in quantum
schemes security arises from the entanglement structure of the joint
state across all participants.

By contrast, this paper considers a model for secret sharing where the dealer distributes the encoded shares to the users by transmitting them across a classical-quantum broadcast channel to the users. This model differs fundamentally from the previously mentioned secret-sharing frameworks where the dealer explicitly determines and distributes each user's exact share. Specifically, in our setting the dealer cannot fully control the share received by any individual user because the shares are perturbed by the broadcast channel. Security is ensured by the noise inherent to the channel.

Our use of channel noise to provide security is similar to classical secret sharing schemes that do not require a secure link between the dealer and the users, which have been considered both in channel-based settings \cite{zou-2015-ClassicalChannelSS,sultana2023secret} and source-based settings \cite{rana-2021-SSCorrelatedRV,chou-2024-DistributedSS,sultana-2025-SecretSharing,chou2019biometric}. In particular, \cite{zou-2015-ClassicalChannelSS} considers secret sharing over a classical broadcast channel, where each set of users (authorized or unauthorized) is treated as a virtual user and achievability is shown via connection to the compound wiretap channel setting \cite{liang-2009-CompoundWTC}. We also leverage this idea of treating a set of users as a virtual user for the purpose of the access structure.

The summary of our contributions is as follows. We consider secret sharing for monotone access structures where a dealer communicates with $L$ users through a classical–quantum broadcast channel, and derive
\begin{itemize}
    \item An achievable one-shot secret sharing rate. This result stems from classical random binning techniques, originally developed for the classical wiretap channel \cite{csiszar-1996-AlmostIndependence}, as well as in  classical-quantum settings, e.g.,~\cite{renes-2011-ChannelCoding,chou-2021-PrivateQMAC}. More specifically, our coding scheme combines source coding with compound side information to handle the reliability constraints, ensuring that authorized sets of users can recover the secret, and privacy amplification with hash functions to handle the security constraints, preventing unauthorized sets of users from gaining information about the secret. 
    The main challenge is in managing multiple authorized sets of users. To this end, we derive one-shot achievability bounds for source coding with compound quantum side information and use them to construct compound channel codes.
    \item An upper bound on one-shot secret sharing rates.  Specifically, we leverage the fact that secret key distribution using a noisy channel is an upper bound on secure communication following \cite{renes-2011-ChannelCoding,qi-2018-PrivateQBrodcast}.
    \item Second-order asymptotics for achievable secret sharing rates. Using existing bounds on the second-order asymptotics of smooth relative entropies from \cite{tomamichel-2013-SecondOrder,li-2014-SecondOrder}, we derive an asymptotic lower bound on the secret sharing rate  for $n$ independent and identically distributed (i.i.d.) uses of the classical-quantum channel.
    \item The secret sharing capacity for classical channels when the access structure requires the participation of all users to reconstruct the secret. In this case, our result reduces to the capacity expression of \cite{zou-2015-ClassicalChannelSS}.
\end{itemize}

The remainder of the paper is organized as follows. Notation and preliminaries are laid out in Section~\ref{sec:notation}. The problem formulation is provided in Section~\ref{sec:prob_statement}, followed by the statement of the main results in Section~\ref{sec:results}. Supporting results needed for the achievability proof are presented in Section~\ref{sec:supporting_results}, followed by the proofs of the one-shot achievability and converse results in Sections~\ref{sec:proof_of_main_achievability_result} and \ref{sec:proof_of_converse}, respectively. The paper is concluded in Section~\ref{sec:conclusion}.

\section{Notation and Preliminaries}\label{sec:notation}

For $\mathcal{H}$, a finite dimensional Hilbert space, let $\mathcal{P}(\mathcal{H})$ be the set of positive semi-definite operators on $\mathcal{H}$. Let $\mathcal{S}_{=}(\mathcal{H}) \triangleq\{\rho\in\mathcal{P}(\mathcal{H}): \Tr(\rho)=1\}$ and $\mathcal{S}_{\leq}(\mathcal{H}) \triangleq\{\rho\in\mathcal{P}(\mathcal{H}): 0\leq\Tr(\rho)\leq1\}$ be the sets of normalized and subnormalized quantum states, respectively. Denote the identity operator on $\mathcal{H}_A$ by $\mathds{1}_A$. Define $\mathcal{H}_{AB}\triangleq\mathcal{H}_A\otimes\mathcal{H}_B$ and, for an operator $L_{AB}$  on $\mathcal{H}_{AB}$,  define $L_A\triangleq \Tr_A[L_{AB}]$.

Define the trace distance between  $\rho \in \mathcal{S}_{\leq}$ and $\sigma \in \mathcal{S}_{\leq}$ as $D(\rho,\sigma)\triangleq\frac{1}{2}\|\rho-\sigma\|_1$, where $\|A\|_1\triangleq\Tr\big[\sqrt{A^\dagger A}\big]$ and $\dagger$ denotes the conjugate transpose. Define the Schatten infinity norm of positive semi-definite operators as $\|A\|=\max_{Q\in\mathcal{S}_\leq(\mathcal{H})}\Tr[AQ]$. Define the fidelity \cite{uhlmann-1976-Fidelity} of two states $\rho,\sigma\in\mathcal{S}_\leq(\mathcal{H})$  as $F(\rho,\sigma)\triangleq\|\sqrt{\rho}\sqrt{\sigma}\|_1^2$ and the purified distance as $P(\rho,\sigma)\triangleq \sqrt{1-F(\rho,\sigma)}$. 

Denote the von Neumann entropy of  $\rho_X\in\mathcal{S}_=(\mathcal{H}_X)$ by $H(X)_\rho$. For $\rho_{AB}\in\mathcal{S}_=(\mathcal{H}_{AB})$, the min- and max-entropy of $A$ conditioned on $B$ for  $\rho_{AB}$ are  given by \cite[Def.~11 and Eq. (6)]{tomamichel-2010-Duality}
\begin{align}
	H_{\min}(A|B)_\rho & \triangleq \max_{\sigma_B\in\mathcal{S}_=(\mathcal{H}_B)}\sup\left\{\lambda\in\mathbb{R}: \rho_{AB}\leq 2^{-\lambda}\mathds{1}_A\otimes\sigma_B\right\},\\
	H_{\max}(A|B)_\rho & \triangleq \max_{\sigma_B\in\mathcal{S}_=(\mathcal{H}_{B})}\log F(\rho_{AB},\mathds{1}_A\otimes\sigma_B). \label{eq:conditional_max_entropy}
\end{align}
For $\epsilon\geq 0$,   define the smooth entropies as \cite[Def.~12 and Lem.~16]{tomamichel-2010-Duality}
\begin{align}
	H_{\min}^\epsilon(A|B)_\rho &\triangleq \max_{\tilde{\rho}_{AB}\in\mathcal{B}_\epsilon(\rho_{AB})}H_{\min}(A|B)_{\tilde{\rho}},\label{eq:smooth_conditional_min_entropy}\\
	H_{\max}^\epsilon(A|B)_\rho & \triangleq \min_{\tilde{\rho}_{AB}\in\mathcal{B}_\epsilon(\rho_{AB})}H_{\max}(A|B)_{\tilde\rho}, \label{eq:smooth_conditional_max_entropy}
\end{align}
where $\mathcal{B}_\epsilon(\rho)\triangleq\{\bar{\rho}\in\mathcal{S}_\leq(\mathcal{H}):P(\rho,\bar{\rho})\leq\epsilon\}$. Note that the definition in \eqref{eq:smooth_conditional_min_entropy} can be equivalently written in terms of the smooth relative max entropy \cite[Def.~3]{tomamichel-2013-SecondOrder},
\begin{align}
    H_{\min}^\epsilon(A|B)_{\rho}& = \max_{\sigma_B\in\mathcal{S}_=(\mathcal{H}_{B})} -D_{\max}^\epsilon(\rho_{AB}\|\mathds{1}_A\otimes\sigma_B),
\end{align}
where $D_{\max}^{\epsilon}(\rho\|\sigma) \triangleq \min_{\tilde{\rho}\in\mathcal{B}_\epsilon(\rho)}\inf\left\{\lambda\in\mathbb{R}: \tilde{\rho}\leq2^\lambda\sigma\right\}$ for $\rho\in\mathcal{S}_=(\mathcal{H})$, $\sigma\in\mathcal{P}(\mathcal{H})$, and $0\leq\epsilon< 1$ \cite[Def.~2]{tomamichel-2013-SecondOrder}. 

For $\rho\in\mathcal{S}_=(\mathcal{H})$ and $\sigma\in\mathcal{P}(\mathcal{H})$, define the quantum relative entropy and relative entropy variance respectively as \cite{tomamichel-2013-SecondOrder,li-2014-SecondOrder}
\begin{align}
    D(\rho\|\sigma) &\triangleq \Tr\left[\rho(\log\rho-\log\sigma)\right],\\
    V(\rho\|\sigma) &\triangleq \Tr\left[\rho(\log\rho-\log\sigma)^2\right] -\left(D(\rho\|\sigma)\right)^2.
\end{align}
For $\rho_{AB}\in\mathcal{S}_=(\mathcal{H}_{AB})$, the quantum conditional entropy and  conditional information variance are respectively given by
\begin{align}
    H(A|B)_\rho &\triangleq -D(\rho_{AB}\|\mathds{1}_A\otimes\rho_B),\\
    V(A|B)_\rho & \triangleq V(\rho_{AB}\|\mathds{1}_A\otimes\rho_B).
\end{align}

The variational distance between two classical probability distributions $p$ and $q$, defined over $\mathcal{X}$, is defined as $\mathbb{V}(q,p)\triangleq\sum_{x\in\mathcal{X}}|p(x)-q(x)|$. $P^U_X$ denotes the uniform distribution on $\mathcal{X}$. Let $[1:x]\triangleq\{1,2,\dots,\lfloor x\rfloor\}$ for $x\in \mathbb{R}$. For a set $\mathcal{A}$, $2^\mathcal{A}$ denotes the power set of $\mathcal{A}$. Denote the indicator function by $\mathbf{1}\{\cdot\}$.

\section{Problem Statement}\label{sec:prob_statement}
Consider the problem of sharing a secret among $L$ users across a classical-quantum broadcast channel $W:\mathcal{X}\rightarrow\mathcal{S}_{=}(\mathcal{H}_{Y_{[1:L]}})$, where $\mathcal{X}$ is a finite set and $\mathcal{H}_{Y_{[1:L]}}=\mathcal{H}_{Y_1}\otimes\cdots\otimes\mathcal{H}_{Y_L}$. The channel maps $x\in\mathcal{X}$ to $\rho_{Y_{[1:L]}}^{x}\in\mathcal{S}_{=}(\mathcal{H}_{Y_{[1:L]}})$, where $\rho_{Y_{[1:L]}}^{x}$ is the state of the quantum system $Y_{[1:L]}$ conditioned on the realization $x$.

The access structure consists of the set of sets $\mathbb{A}\subseteq2^{[1:L]}$ where each element $\mathcal{A}\in\mathbb{A}$ is an \emph{authorized set}, meaning the users in $\mathcal{A}$, when combining their channel outputs, should be able to recover the secret. The access structure $\mathbb{A}$ is \emph{monotone} if $\mathcal{A}\in\mathbb{A}$ and $\mathcal{A}\subseteq\mathcal{A}'\in2^{[1:L]}$ implies $\mathcal{A}'\in\mathbb{A}$ \cite{benaloh-1988-GeneralAccess}. Sets of users not in $\mathbb{A}$ should not be able to recover information about the secret. Accordingly, we define the set of sets $\mathbb{B}\triangleq2^{[1:L]}\backslash \mathbb{A}$, and call each $\mathcal{B}\in\mathbb{B}$ an \emph{unauthorized set}. An example of our setting is depicted in Fig.~\ref{fig:problem_setting}.

\begin{figure}[t]
    \centering
    \resizebox{0.6\linewidth}{!}{
    \begin{tikzpicture}
     \node (0) at (2.25*\x,1.8*\y) {Distribution};
     \node (a) at (-0.1*\x,0.0) {$S$};
     \node (b) at (1.05*\x,0.0) [draw,rounded corners = 2pt, minimum width=0.9cm,minimum height=0.5cm, align=center] {$\Enc$};
     \node (c) at (2.35*\x,0.0) [draw,rounded corners = 2pt, minimum width=0.9cm,minimum height=0.5cm, align=center] {$W$};
     \node (d) at (4.5*\x,1.2*\y) [bob,minimum size=\b cm] {User $1$};
     \node (e) at (4.5*\x,0) [bob,minimum size=\b cm] {User $2$};
     \node (f) at (4.5*\x,-1.2*\y) [bob,minimum size=\b cm] {User $3$};
     \draw[decoration={markings,mark=at position 1 with {\arrow[scale=1.5]{latex}}},
      postaction={decorate}, thick, shorten >=1.4pt] ($(a.east)+(0.0,0)$) -- ($(b.west)-(0,0)$);
     \draw[decoration={markings,mark=at position 1 with {\arrow[scale=1.5]{latex}}},
      postaction={decorate}, thick, shorten >=1.4pt] ($(b.east)+(0.0,0)$) -- ($(c.west)-(0,0)$) node [midway, above] {$X$};
     \draw[decoration={markings,mark=at position 1 with {\arrow[scale=1.5]{latex}}},
     postaction={decorate}, thick, shorten >=1.4pt] ($(c.east)+(0.0,0.15)$) -- ($(d.west)-(0,0)$) node [midway, above left] {$\rho_{Y_1}^x$};
     \draw[decoration={markings,mark=at position 1 with {\arrow[scale=1.5]{latex}}},
     postaction={decorate}, thick, shorten >=1.4pt] ($(c.east)+(0.0,0)$) -- ($(e.west)-(0,0)$) node [midway, above right] {$\rho_{Y_2}^x$};
     \draw[decoration={markings,mark=at position 1 with {\arrow[scale=1.5]{latex}}},
     postaction={decorate}, thick, shorten >=1.4pt] ($(c.east)+(0.0,-0.15)$) -- ($(f.west)-(0,0)$) node [midway, below left] {$\rho_{Y_3}^x$};

     \node (1) at (2.25*\x,-2.15*\y) {Recovery};
     \node (aa) at (.1*\x,-3*\y) [bob,minimum size=\b cm] {};
     \node (aa1) at (.2*\x,-3*\y) [bob,minimum size=\b cm] {};
     \node (aa2) at (.3*\x,-3*\y) [bob,minimum size=\b cm] {};
     \node (aa3) at (1.25*\x,-3*\y) {$S$};
     \node (aal) at (.2*\x,-3.7*\y) {$\{1,2,3\}$};
     \node (bb) at (2*\x,-3*\y) [bob,minimum size=\b cm] {};
     \node (bb1) at (2.1*\x,-3*\y) [bob,minimum size=\b cm] {};
     \node (bb2) at (2.95*\x,-3*\y) {$S$};
     \node (bbl) at (2.05*\x,-3.7*\y) {$\{1,2\}$};
     \node (cc) at (3.65*\x,-3*\y) [bob,minimum size=\b cm] {};
     \node (cc1) at (3.75*\x,-3*\y) [bob,minimum size=\b cm] {};
     \node (cc2) at (4.6*\x,-3*\y) {$S$};
     \node (ccl) at (3.7*\x,-3.7*\y) {$\{2,3\}$};
     \node (dd) at (-.15*\x,-4.5*\y) [bob,minimum size=\b cm] {};
     \node (dd1) at (-0.05*\x,-4.5*\y) [bob,minimum size=\b cm] {};
     \node (dd2) at (.75*\x,-4.5*\y) {$\cancel S$};
     \node (ddl) at (-.1*\x,-5.2*\y) {$\{1,3\}$};
     \node (ee) at (1.35*\x,-4.5*\y) [bob,minimum size=\b cm] {};
     \node (ee1) at (2.15*\x,-4.5*\y) {$\cancel S$};
     \node (eel) at (1.35*\x,-5.2*\y) {$\{1\}$};
     \node (ff) at (2.75*\x,-4.5*\y) [bob,minimum size=\b cm] {};
     \node (ff1) at (3.55*\x,-4.5*\y) {$\cancel S$};
     \node (ffl) at (2.75*\x,-5.2*\y) {$\{2\}$};
     \node (gg) at (4.15*\x,-4.5*\y) [bob,minimum size=\b cm] {};
     \node (gg1) at (4.95*\x,-4.5*\y) {$\cancel S$};
     \node (ggl) at (4.15*\x,-5.2*\y) {$\{3\}$};
     
     \draw[decoration={markings,mark=at position 1 with {\arrow[scale=1.5]{latex}}}, postaction={decorate}, thick, shorten >=1.4pt] ($(aa2.east)+(0.0,0)$) -- ($(aa3.west)-(0,0)$);
     \draw[decoration={markings,mark=at position 1 with {\arrow[scale=1.5]{latex}}}, postaction={decorate}, thick, shorten >=1.4pt] ($(bb1.east)+(0.0,0)$) -- ($(bb2.west)-(0,0)$);
     \draw[decoration={markings,mark=at position 1 with {\arrow[scale=1.5]{latex}}}, postaction={decorate}, thick, shorten >=1.4pt] ($(cc1.east)+(0.0,0)$) -- ($(cc2.west)-(0,0)$);
     \draw[decoration={markings,mark=at position 1 with {\arrow[scale=1.5]{latex}}}, postaction={decorate}, thick, shorten >=1.4pt] ($(dd1.east)+(0.0,0)$) -- ($(dd2.west)-(0,0)$);
     \draw[decoration={markings,mark=at position 1 with {\arrow[scale=1.5]{latex}}}, postaction={decorate}, thick, shorten >=1.4pt] ($(ee.east)+(0.0,0)$) -- ($(ee1.west)-(0,0)$);
     \draw[decoration={markings,mark=at position 1 with {\arrow[scale=1.5]{latex}}}, postaction={decorate}, thick, shorten >=1.4pt] ($(ff.east)+(0.0,0)$) -- ($(ff1.west)-(0,0)$);
     \draw[decoration={markings,mark=at position 1 with {\arrow[scale=1.5]{latex}}}, postaction={decorate}, thick, shorten >=1.4pt] ($(gg.east)+(0.0,0)$) -- ($(gg1.west)-(0,0)$);
    \end{tikzpicture}
    }
    \caption{A secret $S$ is shared among three users over a classical-quantum broadcast channel $W$ with access structure $\mathbb{A}=\big\{\{1,2,3\},\{1,2\},\{2,3\}\big\}$. Only sets of users in the authorized set can recover the secret.}
    \label{fig:problem_setting}
\end{figure}
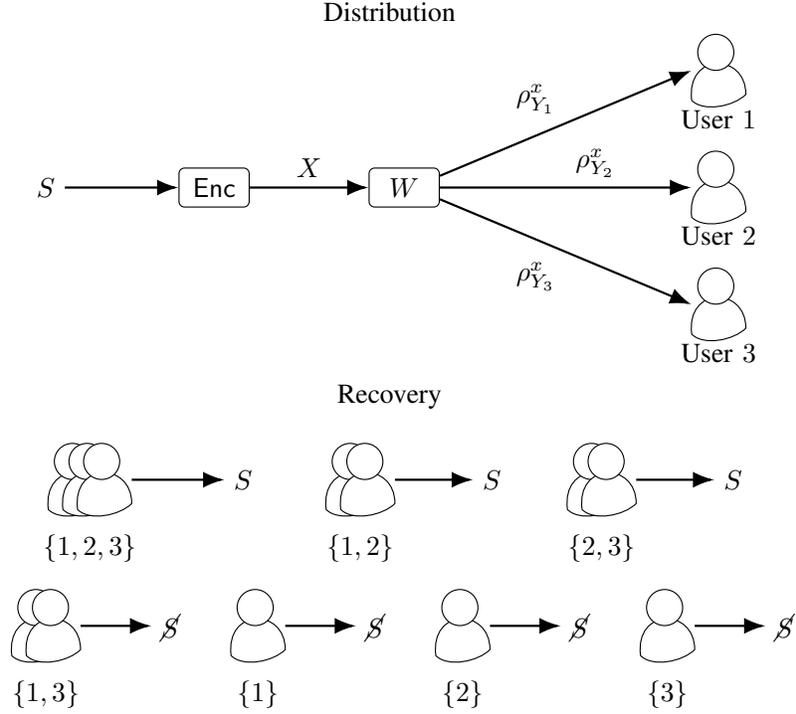

\subsection{One-shot Regime}

The one-shot regime considers a single use of the channel.

\begin{definition} \label{def:channel_code}
    A $2^R$-code for secret sharing among $L$ users with monotone access structure $\mathbb{A}$ over a classical-quantum broadcast channel $W$ consists of 
    \begin{itemize}
        \item a secret set $\mathcal{S} \triangleq [1:2^{R}]$;
        \item an encoder $\Enc: \mathcal{S}\rightarrow\mathcal{X}$;
        \item $|\mathbb{A}|$ decoding maps $\Dec_{\mathcal{A}}: \mathcal{H}_{Y_{\mathcal{A}}}\rightarrow\mathcal{S}$, where $\mathcal{A}\in\mathbb{A}$;
    \end{itemize}
    and operates as follows. The transmitter encodes the secret $S$ as $\Enc(S)$, and transmits the result across the classical-quantum channel $W$, from which User $l\in[1:L]$ receives $\rho_{Y_l}^{\Enc(S)}$. An authorized set $\mathcal{A}\in\mathbb{A}$  estimate $S$ as $\Dec_\mathcal{A}(\rho_{Y_{\mathcal{A}}}^{\Enc(S)})$, where $(\rho_{Y_{\mathcal{A}}}^{\Enc(S)})\triangleq (\rho_{Y_{l}}^{\Enc(S)})_{l\in \mathcal{A}}$.
\end{definition}

\begin{definition}\label{def:nonasymptotic_epsilon_good}
    A $2^R$-code is $\epsilon$-good if 
\begin{align}\label{eq:definition_nonasymptotic_reliability_condition}
    \max_{\mathcal{A}\in\mathbb{A}}\Pr\big\{S\neq \Dec_\mathcal{A}(\rho_{Y_{\mathcal{A}}}^{\Enc(S)})\big\} &\leq \epsilon,  \qquad\qquad\textnormal{(Reliability)}\\
    \max\limits_{\mathcal{B}\in\mathbb{B}}\big\|\tilde{\rho}_{SY_\mathcal{B}}-\rho_S\otimes\sigma_{Y_\mathcal{B}}\big\|_1&\leq\epsilon,  \qquad\qquad\textnormal{(Security) }\label{eq:definition_non_asymptotic_security_condition}
\end{align}
where $\rho_{S}$ is the maximally mixed state on $\mathcal{H}_S$,  $\sigma_{Y_\mathcal{B}}$ is an arbitrary quantum state on $\mathcal{H}(Y_\mathcal{B})$, and $\Tilde{\rho}_{SY_\mathcal{B}}\triangleq \sum_{s\in \mathcal{S}} \frac{1}{|\mathcal{S}|} |s\rangle \langle s| \otimes \rho_{Y_{\mathcal{B}}}^{\Enc(s)}$.
\end{definition}

\subsection{Asymptotic Regime}
The asymptotic regime considers $n$ independent uses of the channel $W$, as $n$ goes to infinity. Consider a $(2^{nR},n)$-code for the channel $W^{\otimes n}$. 

\begin{definition} \label{def:asypmtotic_reliability}
    A rate $R$ is  achievable if there exists a sequence of $(2^{nR},n)$-codes satisfying
\begin{align}
    \lim\limits_{n\rightarrow\infty}\max_{\mathcal{A}\in\mathbb{A}}\Pr\big\{S\neq \Dec_\mathcal{A}(\rho_{Y^n_{\mathcal{A}}}^{\Enc(S)})\big\}&=0 , \ \qquad\qquad\textnormal{(Reliability) }\\
    \lim\limits_{n\rightarrow\infty}\max\limits_{\mathcal{B}\in\mathbb{B}}\big\|\tilde{\rho}_{SY^n_\mathcal{B}}-\rho_S\otimes\sigma_{Y^n_\mathcal{B}}\big\|_1 & = 0. \qquad\qquad\textnormal{(Security)} \label{eq:def_asypmtotic_security}
\end{align}
\end{definition}

\section{Results} \label{sec:results}
We first present our main result, a one-shot secret sharing achievability bound.

\begin{theorem} \label{thm:general_one_shot_secret_sharing_ach}
    For secret sharing over a classical-quantum broadcast channel $W$ among $L$ users with monotone access structure $\mathbb{A}$, there exists an $\epsilon$-good $2^R$-code satisfying 
    \begin{align}
        R \geq \max_{P_X}\Big[\min_{\mathcal{B\in\mathbb{B}}}H^{\epsilon'}_\textnormal{min}(X|Y_\mathcal{B})_\psi - \max_{\mathcal{A}\in\mathbb{A}} H_\textnormal{max}^{\epsilon_1}(X|Y_\mathcal{A})_\psi - \delta -2\log\frac{1}{\epsilon_2} - 2\log{(|\mathbb{A}|\!+\!1)} - 6\Big], \label{eq:one_shot_lower_bound}
    \end{align}
    where $\epsilon_1,\epsilon_2,\delta>0$, $\epsilon'=\epsilon_1(|\mathbb{A}|+1)+\epsilon_2$, and $\epsilon = |\mathbb{A}|(3\epsilon'+2^{-\delta/2-1}) + |\mathbb{B}|(4\epsilon' + 2^{-\delta/2})$, and 
    \begin{align}
        \psi_{XY_{[1:L]}}\triangleq\sum_{x\in\mathcal{X}}P_X(x)|x\rangle\langle x|\otimes\rho_{Y_{[1:L]}}^x. \label{eq:classical_quantum_state}
    \end{align}
\end{theorem}
\begin{IEEEproof}
    See Section~\ref{sec:proof_of_main_achievability_result}.
\end{IEEEproof}

The minimum over unauthorized sets captures the worst-case privacy amplification constraint, while the maximum over authorized sets represents the rate penalty for source coding with quantum side information for the worst-case authorized set.
Additionally,  $\delta$ embodies a trade-off between $\epsilon$ and the rate. Notice that as $\delta$ increases, $\epsilon$ decreases, but the achievable rate decreases.

We have the following upper bound on  achievable rates.

\begin{theorem} \label{thm:general_asymptotic_secret_sharing_converse}
    For secret sharing over a classical-quantum broadcast channel $W$ among $L$ users with monotone access structure $\mathbb{A}$, the rate $R$ of an $\epsilon$-good code is upper bounded~by 
    \begin{align}
        R \leq \max_{P_{SX}}\left[\min_{\mathcal{B}\in\mathbb{B}}H_\textnormal{min}^{\sqrt{\epsilon}}(S|Y_\mathcal{B})_\psi-\max_{\mathcal{A}\in\mathbb{A}}H_\textnormal{max}^{\sqrt{2\epsilon}}(S|Y_\mathcal{A})_\psi\right],
    \end{align}
    where
    \begin{align}
        \psi_{SXY_{[1:L]}}\triangleq\sum_{x\in\mathcal{X}}P_{SX}(s,x)|s\rangle\langle s|\otimes|x\rangle\langle x|\otimes\rho_{Y_{[1:L]}}^x.
    \end{align}
\end{theorem}
\begin{IEEEproof}
    See  Section~\ref{sec:proof_of_converse}.
\end{IEEEproof}

This converse follows by considering an upper bound on secret key distribution over the channel.

When sharing a secret across   $n$ i.i.d. uses of a classical-quantum broadcast channel, we have the following second order characterization of the rate.

\begin{theorem} \label{thm:general_second_order_asymptotics}
   For secret sharing over a classical-quantum broadcast channel $W$ among $L$ users with monotone access structure $\mathbb{A}$, there exists a sequence of $(2^{nR},n)$-codes with achievable rate $R$ satisfying
    \begin{align}
        R &\geq \max_{P_X} \Bigg\{\min_{\mathcal{B\in\mathbb{B}}}\left[H(X|Y_\mathcal{B})_\psi + \sqrt{\frac{1}{n}V(X|Y_\mathcal{B})_\psi}\Phi^{-1}({\epsilon'}^2)\right]\nonumber\\
        & \qquad\qquad\qquad- \max_{\mathcal{A}\in\mathbb{A}}\left[H(X|Y_\mathcal{A})_\psi - \sqrt{\frac{1}{n}V(X|Y_\mathcal{A})_\psi}\Phi^{-1}(\epsilon_1)\right]+\mathcal{O}\left(\frac{\log n}{n}\right)\Bigg\}, \label{eq:second_order_expansion_result}
    \end{align}
    where   $\psi_{XY_{[1:L]}}$ is defined in \eqref{eq:classical_quantum_state}, $\epsilon_1>0$, $\delta>0$, $\epsilon'=\epsilon_1(|\mathbb{A}+1|)$, $\epsilon$ is defined as in Theorem~\ref{thm:general_one_shot_secret_sharing_ach}, and $\Phi^{-1}$ is the inverse cumulative distribution function of a standard normal random variable.
\end{theorem}
\begin{IEEEproof}
    See Appendix~\ref{app:second_order_asymptotic_proof}.
\end{IEEEproof}

After accounting for decay in the probability of error as $n\rightarrow\infty$, the next result   follows  from  Theorem \ref{thm:general_second_order_asymptotics}.
\begin{corollary}\label{cor:asymptotic_main_result}
   For secret sharing over a classical-quantum broadcast channel $W$ among $L$ users with monotone access structure $\mathbb{A}$, there exists a sequence of $(2^{nR},n)$-codes with achievable rate $R$ satisfying
    \begin{align}
        R \geq \max_{P_X}\Big[\min_{\mathcal{B\in\mathbb{B}}}H(X|Y_\mathcal{B})_\psi - \max_{\mathcal{A}\in\mathbb{A}} H(X|Y_\mathcal{A})_\psi\Big],
    \end{align}
     where   $\psi_{XY_{[1:L]}}$ is defined in \eqref{eq:classical_quantum_state}.
\end{corollary}
\begin{IEEEproof}
    See Appendix~\ref{app:proof_asymptotic_general_secret_sharing_achievability_proof}.
\end{IEEEproof}

We now consider the special case of an access structure in which all shares are required to recover the secret, i.e., $\mathbb{A}=\{[1:L]\}$. In this case, we do not need to consider source coding with \emph{compound} side information, allowing us to recover a better bound (in addition to the bound no longer depending on the size of the authorized set).
\begin{corollary}\label{cor:nonasymptotic_all_user_access_structure}
    For secret sharing over a classical-quantum broadcast channel $W$ among $L$ user with the access structure $\mathbb{A}=\{[1:L]\}$, there exists an $\epsilon$-good $2^R$-code  satisfying 
    \begin{align}
        R \geq \max_{P_X}\Big[\min_{\mathcal{B\in\mathbb{B}}}H^{\epsilon'}_\textnormal{min}(X|Y_\mathcal{B})_\psi -  H_\textnormal{max}^{\epsilon_1}\big(X|Y_{[1:L]}\big)_\psi - \delta -2\log\frac{1}{\epsilon_2} - 6\Big].
    \end{align}
    for $\epsilon_1,\epsilon_2,\delta>0$, $\epsilon=3\epsilon'+2^{-\delta/2-1} + |\mathbb{B}|(4\epsilon' + 2^{-\delta/2})$, and  $\psi_{XY_{[1:L]}}$  defined in \eqref{eq:classical_quantum_state}.
\end{corollary}
\begin{IEEEproof}
    See  Appendix~\ref{app:proof_nonasymptotic_all_user_access_structure_achievability_proof}.
\end{IEEEproof}

We can also consider the asymptotic characterization for the access structure $\mathbb{A}=\{[1:L]\}$.

\begin{corollary}\label{cor:asymptotic_all_user_access_structure}
 For secret sharing over a classical-quantum broadcast channel $W$ among $L$ users with the access structure $\mathbb{A}=\{[1:L]\}$, there exists a sequence of $(2^{nR},n)$-codes with achievable rate $R$ satisfying
    \begin{align}
        R \geq \max_{P_X}\Big[\min_{\mathcal{B\in\mathbb{B}}}H(X|Y_\mathcal{B})_\psi -  H(X|Y_{[1:L]})_\psi\Big],
    \end{align}
    where   $\psi_{XY_{[1:L]}}$ is defined in \eqref{eq:classical_quantum_state}.
\end{corollary}
\begin{IEEEproof}
    The result follows by simplifying Corollary~\ref{cor:asymptotic_main_result} using $\mathbb{A} =\{[1:L]\}$.
\end{IEEEproof}

Consider now the asymptotic regime where $W$ is classical. We define the secret sharing capacity as the maximum rate at which a secret can be shared reliably and securely across $W$.
\begin{corollary}[Classical Capacity] \label{cor:classical_secret_sharing_capacity_all_user}
    The secret sharing capacity over a classical broadcast channel among $L$ users with the access structure $\mathbb{A}=\{[1:L]\}$ is
    \begin{align}
\max\limits_{P_{X}}\min_{\mathcal{B}\in\mathbb{B}}I(X;Y_{[1:L]\setminus\mathcal{B}}|Y_\mathcal{B}).
    \end{align}
\end{corollary}
\begin{IEEEproof}
    See Appendix~\ref{app:proof_classical_capacity}.
\end{IEEEproof}

This result coincides with the  capacity result for secret sharing over a classical broadcast channel in \cite{zou-2015-ClassicalChannelSS}.

\section{Supporting Results} \label{sec:supporting_results}

In this section, we state supporting results that will be used in the proof of Theorem~\ref{thm:general_one_shot_secret_sharing_ach}.

\subsection{Hash Functions}\label{sec:hash_functions_supporting_results}
We will leverage properties of 2-universal hash functions, including the leftover hash lemma, to  obtain a result for one-shot source coding with compound quantum side information, and to provide uniformity and security guarantees in the proof of Theorem~\ref{thm:general_one_shot_secret_sharing_ach}.

\begin{definition}[2-Universal Hash Family, \cite{carter-1977-UniversalHash}] \label{def:hash_function}
    The family of hash functions $\mathcal{F}$, consisting of hash functions $f: \mathcal{X}\rightarrow \{0,1\}^r$, is 2-universal if $\forall x,x' \in \mathcal{X}, x\neq x' \implies \Pr[F(x)=F(x')]\leq 2^{-r}$, where $F$ is uniformly selected from $\mathcal{F}$.
\end{definition} 

Let $\Gamma\in \big[1:|\mathcal{F}|\big]$ index the hash functions in the 2-universal hash family $\mathcal{F}$, and  define the classical-quantum state
\begin{align}
\rho_{f_\Gamma(X)\Gamma Z} & \triangleq \sum\limits_{u}\sum\limits_{\gamma} \sum\limits_{x} P_{f_\Gamma(X)\Gamma X}(u,\gamma,x)|u\rangle\langle u|\otimes |\gamma\rangle\langle \gamma|\otimes\rho_{Z}^{x},
\end{align}
where $P_{f_\Gamma(X)\Gamma X}$ is the distribution induced by the family of 2-universal hash functions for arbitrary $P_X$ and independent uniformly distributed $\Gamma$, i.e., 
\begin{align}
    P_{f_\Gamma(X)\Gamma X}(u,\gamma, x) = P_\Gamma^U(\gamma)P_X(x)\mathds{1}\{f_\gamma(x)=u\}. \label{eq:probability_distribution_induced_by_the_hash_function}
\end{align}

\begin{lemma}[Leftover Hash Lemma, \cite{tomamichel-2011-LeftoverHashing}]\label{lem:lhl}
For any $\epsilon>0$, $\rho_{f_\Gamma(X)\Gamma Z}$ defined above, and the maximally mixed state $\rho_U$ on $\mathcal{H}_{f_\Gamma(X)}$, we have
    \begin{align}
        \|\rho_{f_\Gamma(X)\Gamma Z} - \rho_U\otimes\rho_{\Gamma Z} \|_1\leq 2 \epsilon + \sqrt{2^{r-H^\epsilon_{\min}(X|Z)_\rho}}.
    \end{align}
\end{lemma}

\subsection{Source Coding with Compound Quantum Side Information}\label{sec:source_coding_supporting_results}

Consider the classical quantum state $\psi_{XB^i}=\sum_{x\in\mathcal{X}}P(x)|x\rangle\langle x|\otimes\varphi^x_{B^i}$,   $i\in[1:k]$, where we define $\varphi^x_{B^i}\in \mathcal{S}_=(\mathcal{H}_{B^i})$ to be the density operator of system $B^i$ given $X=x$. The transmitter holds the classical random variable $X\in\mathcal{X}$ and the receiver observes one of the $k$ quantum systems $B^i$. The transmitter encodes $X$ into $C\in\mathcal{C}$, which is noiselessly communicated to the receiver, who attempts to recover $X$. The alphabets $\mathcal{X}$ and $\mathcal{C}$ are assumed to be finite. 

\begin{definition} \label{def:source_code_compound_quantum_side_information}
    A source code with compound quantum side information for  the sequence of classical quantum state $\{\psi_{XB^1},\dots,\psi_{XB^k}\}$ consists of 
    \begin{itemize}
        \item an encoder $g:\mathcal{X}\rightarrow\mathcal{C}$;
        \item decoders $h_i:\mathcal{S}_=(\mathcal{H}_{B^i})\times\mathcal{C}\rightarrow\mathcal{X}$, for each $i\in[1:k]$;
    \end{itemize}
   has probability of error defined by
    \begin{align}
        p_e=\max_{i\in[1:k]} \Pr\Big[h_i\big(\varphi^X_{B^i},g(X)\big)\neq X\Big],
    \end{align}
    and is said to be $\epsilon$-good if $p_e\leq\epsilon$.
\end{definition}

\begin{lemma}\label{lem:one-shot-source-code-compound-quantum-side-information}
    Let $\epsilon>0$. For a sequence of classical-quantum states $\{\psi_{XB^1},\dots,\psi_{XB^k}\}$, there exists an $\epsilon$-good source code with compound quantum side information satisfying 
    \begin{equation}
        \log|\mathcal{C}|=\max\limits_{i\in[1:k]} \Big\lceil H_\textnormal{max}^{\epsilon_1}(X|B^i)_\psi-2\log\epsilon_2 + 2\log{(k+1)}\Big\rceil + 3,
    \end{equation}
    where $\epsilon_1,\epsilon_2\geq0$ and $\epsilon_1+\epsilon_2/(k+1)=\epsilon/(k+1)$.
\end{lemma}
\begin{IEEEproof}
    See Appendix~\ref{app:proof_of_lem_source_coding_with_compound_quantum_side_information}.
\end{IEEEproof}
    
Our proof extends the one-shot source coding result with quantum side information in \cite[Th.~1]{renes-2012-OneShotCompression}. Building on the argument therein, we derive a source coding result for compound quantum side information. Under the assumption that the set of possible side-information states is finite, we show that a single encoder can be constructed that is reliable for all side-information states.

\subsection{Channel Coding from Source Coding with Compound Quantum Side Information} \label{sec:channel_coding_supporting_results}
Reference \cite{renes-2011-ChannelCoding} has shown that a channel code can be constructed from a protocol for source coding with quantum side information. We establish a similar result for the compound setting, relating channel coding for a compound channel to source coding with compound quantum side information.
\begin{lemma} \label{lem:information_reconcilliation_to_channel_coding}
    Given the classical-quantum channel $\Theta:\mathcal{U}\rightarrow \mathcal{S}_=(\mathcal{H}_{Y_{[1:L]}})$ with almost uniformly distributed input $U\sim \tilde{P}_{U}$, i.e., $\mathbb{V}(\tilde{P}_{U},P_{U}^U)\leq\epsilon'$,  suppose that we have an $\epsilon''$-reliable source code for source $U$ with compound quantum side information $(\rho_{Y_\mathcal{A}}^{U})_{\mathcal{A}\in\mathbb{A}}$, with relationship defined by $\Theta$, consisting of a linear encoder $g:\mathcal{U}\rightarrow\mathcal{C}$ and  decoders $h_\mathcal{A}:\mathcal{S}_=(\mathcal{H}_{Y_\mathcal{A}})\times\mathcal{C}\rightarrow\mathcal{U}$ for all $\mathcal{A}\in\mathbb{A}$.
    Then, there exists a family of encoders and decoders $\big\{\Enc_c,\{\Dec_{c,\mathcal{A}}\}_{\mathcal{A}\in\mathbb{A}}\big\}_{c\in\mathcal{C}}$, where $\Enc_c:\mathcal{S}\rightarrow\mathcal{U}$ and $\Dec_{c,\mathcal{A}}:\mathcal{H}_{Y_\mathcal{A}}\rightarrow\mathcal{S}$ for a set of uniformly distributed secrets $\mathcal{S}$ with $|\mathcal{S}| = |\mathcal{U}|/|\mathcal{C}|$, for which the probability of error averaged over uniform $C$ is bounded by $\epsilon'+\epsilon''$, i.e., 
    \begin{align}
        \max_{\mathcal{A}\in\mathbb{A}}\sum\limits_{c\in\mathcal{C}}\frac{1}{|\mathcal{C}|}\bar{P}_e^{c,\mathcal{A}} \leq \epsilon'+\epsilon'',
    \end{align}
    where $\bar{P}_e^{c,\mathcal{A}}\triangleq \mathbb{E}\Big[\Pr\big[\Dec_{c,\mathcal{A}}\big(\rho_{Y_\mathcal{A}}^{\Enc_c(S)}\big)\neq S\big]\Big]$ and the expectation is taken over $S$.
    
\end{lemma}
\begin{IEEEproof}
    See Appendix~\ref{app:proof_of_information_reconcilliation_to_channel_coding}.
\end{IEEEproof}

\section{Proof of Theorem~\ref{thm:general_one_shot_secret_sharing_ach}}\label{sec:proof_of_main_achievability_result}

\subsection{Overview}\label{sec:proof_overview}
The construction of the coding scheme follows the structure of \cite{renes-2011-ChannelCoding}, and is inspired by random binning \cite{yassaee-2014-OSRB,csiszar-1996-AlmostIndependence}, which leverages source coding with side information and privacy amplification via hash functions. The key insight is that a source coding protocol can be used as a building block for a channel coding scheme (Lemma~\ref{lem:information_reconcilliation_to_channel_coding}), provided the channel input is uniformly distributed. Since the optimal input distribution is generally not uniform, we create a function, referred to as {\it shaper}, using  2-universal hash functions, which maps an almost uniform random variable to one with an arbitrary distribution. The leftover hash lemma (Lemma~\ref{lem:lhl}) is the key tool to both approximate the desired input distribution from a uniform distribution \textit{and} secure the secret.  Defining a virtual channel $W'$ as the composition of the shaper and the channel $W$, gives us a channel with almost uniform input. We then design a source coding protocol for the virtual channel $W'$ (Lemma~\ref{lem:one-shot-source-code-compound-quantum-side-information}),  which is then used to construct a channel coding scheme for $W$. 

\subsection{Coding Scheme}\label{sec:coding_scheme}

First, we create a shaper that will take an almost uniformly distributed random variable as input and outputs a random variable distributed according to $P_X$. We do this using a family of 2-universal hash functions with function index $\Gamma$, input $X\sim P_X$, and output $U\in\mathcal{U}$. The hash function induces the distribution \eqref{eq:probability_distribution_induced_by_the_hash_function} from which we obtain the conditional distribution
\begin{align}
    P_{X|f_\Gamma(X)\Gamma}(x|u,\gamma) = \frac{P_{f_\Gamma(X)\Gamma X}(u,\gamma,x)}{\sum\limits_{x\in\mathcal{X}}P_{f_\Gamma(X)\Gamma X}(u,\gamma,x)},\label{eq:shaper_conditional_distribution}
\end{align}
and $U\sim P_{f_\Gamma(X)}$ with $P_{f_\Gamma(X)}(u)=\sum_{x}\sum_\gamma P_{f_\Gamma(X)\Gamma X}(u,\gamma,x)$.  
We define the shaper $\Shp: \mathcal{U}\times \Gamma\rightarrow \mathcal{X}$ by the conditional distribution \eqref{eq:shaper_conditional_distribution} such that the shaper induces the joint distribution
\begin{align}
    \tilde{P}_{U\Gamma X}(u,\gamma,x) =P_{f_\Gamma(X)}(u)P^U(\gamma)P_{X|f_\Gamma(X)\Gamma}(x|u,\gamma). \label{eq:shaper_distribution}
\end{align}
We then define the virtual channel $W'$ as the concatenation of the shaper and the classical quantum channel $W$, i.e., $W'(U,\gamma)=W(\Shp(U,\Gamma))$. Since the hash function index $\Gamma$ is independent of the virtual channel input $U$ and is known to all parties, we consider it to be part of the channel and omit it from the arguments of $W'$ and write $W'(U)$.

Then, using Lemma~\ref{lem:one-shot-source-code-compound-quantum-side-information}, we obtain an $\epsilon'$-good source code with compound quantum side information for the sequence of classical-quantum states $\{\psi_{U\mathcal{A}}\}_{\mathcal{A}\in\mathbb{A}}$, where 
\begin{align}
    \psi_{U\mathcal{A}}=\sum_{u\in\mathcal{U}}\sum_{\gamma\in\mathit{\Gamma}}\sum_{x\in\mathcal{X}}P^U(u)P^U(\gamma)P_{X|f_\Gamma(X)\Gamma}(x|u,\gamma)|u\rangle\langle u| \otimes \varphi^x_\mathcal{A},
\end{align}
and $\varphi_\mathcal{A}^x$ is the joint state received by the users in $\mathcal{A}$ when the input to $W$ is $x$.
Using Lemma~\ref{lem:information_reconcilliation_to_channel_coding}, we then construct a family of encoders/decoders $\big\{\Enc_c,\{\Dec_{c,\mathcal{A}}\}_{\mathcal{A}\in\mathbb{A}}\big\}_{c\in\mathcal{C}}$, where $\Enc_c:\mathcal{S}\rightarrow\mathcal{U}$ and $\Dec_{c,\mathcal{A}}:\mathcal{H}_{Y_\mathcal{A}}\rightarrow\mathcal{S}$, using the source code with compound quantum side information for the virtual channel $W'$.

In Section~\ref{sec:encoder_decoder_selection}, a specific value $c^*\in\mathcal{C}$ is selected to satisfy the reliability and security requirements, such that we obtain the encoder and decoders for the channel $W$ by defining $\Enc(S)=\Shp(\Gamma,\Enc_{c^*}(S))$, where $\Gamma$ is selected uniformly and independent from $S$, and considering the decoders $\{\Dec_{c^*,\mathcal{A}}\}_{\mathcal{A}\in\mathbb{A}}$.

\subsection{Analysis}
In this section, we find the rate conditions that guarantee the reliability \eqref{eq:definition_nonasymptotic_reliability_condition} and security \eqref{eq:definition_non_asymptotic_security_condition} of the coding scheme described in Subsection~\ref{sec:coding_scheme}, and deduce an achievable rate.

\subsubsection{Reliability}
By Lemma~\ref{lem:one-shot-source-code-compound-quantum-side-information}, the  $\epsilon'$-good source code with compound quantum side information  for the sequence of classical-quantum states $\{\psi_{U\mathcal{A}}\}_{\mathcal{A}\in\mathbb{A}}$  can be chosen such that the compression alphabet $\mathcal{C}$ satisfies
\begin{align}
    \log|\mathcal{C}|=\max\limits_{\mathcal{A}\in\mathbb{A}} \Big\lceil H_\text{max}^{\epsilon_1}\big(U|Y_\mathcal{A}\big)_\psi-2\log\epsilon_2 + 2\log{(|\mathbb{A}|+1)}\Big\rceil + 3, \label{eq:compression_alphabet_lower_limit}
\end{align}
where $\epsilon_1,\epsilon_2>0$ and $\epsilon_1(|\mathbb{A}|+1)+\epsilon_2 = \epsilon'$. 
Then, we bound the distance between the uniform distribution and the shaper input distribution $P_{f_\Gamma(X)}$ as
\begin{align}
    \mathbb{V}(P_{f_\Gamma(X)},P^U_{U}) & \overset{(a)}{\leq} \mathbb{V}\big(P_{f_\Gamma(X)\Gamma X},P_U^UP_{\Gamma}^UP_{X}\big)\\
    & \overset{(b)}{\leq} 2\epsilon'+\sqrt{2^{\log|\mathcal{U}|-H^{\epsilon'}_\text{min}(X)_\psi}},
\end{align}
where $(a)$ is due to Lemma~\ref{lem:variational_distance_marginal_smaller} in Appendix \ref{app_I} and $(b)$ is due to  Lemma~\ref{lem:lhl}, which is small if 
\begin{align}
    \log|U| < H^{\epsilon'}_\text{min}(X)_\psi. \label{eq:shaper_uniformity_condition}
\end{align}

Finally, we can evaluate the average reliability of the family of encoder and decoders constructed by Lemma~\ref{lem:information_reconcilliation_to_channel_coding} for the channel $W'$.
Specifically, $\big\{\Enc_c,\{\Dec_{c,\mathcal{A}}\}_{\mathcal{A}\in\mathbb{A}}\big\}_{c\in\mathcal{C}}$ created in Lemma~\ref{lem:information_reconcilliation_to_channel_coding} satisfy
\begin{align}
    \max_{\mathcal{A}\in\mathbb{A}}\sum\limits_{c\in\mathcal{C}}\frac{1}{|\mathcal{C}|}\bar{P}_e^{c,\mathcal{A}} \leq 3\epsilon'+\sqrt{2^{\log|\mathcal{U}|-H^{\epsilon'}_\text{min}(X)_\psi}},\label{eq:probability_of_error_wprime}
\end{align}
for some set of secrets $\mathcal{S}$ satisfying
\begin{align}
    |\mathcal{S}|=|\mathcal{U}|/|\mathcal{C}|.    \label{eq:alphabet_constraints}
\end{align}

\subsubsection{Security}
We now find conditions under which the security constraints are satisfied for an encoder  selected uniformly at random from $\{\Enc_c\}_{c\in\mathcal{C}}$ . The following holds for each unauthorized set $\mathcal{B}\in\mathbb{B}$. The  classical quantum state induced by the channel coding protocol presented above for a uniformly selected encoder is 
\begin{align}
    \bar{\rho}_{SU\Gamma Y_\mathcal{B}} = |\mathcal{C}|^{-1} \sum_c \tilde{\rho}^c_{SU\Gamma Y_\mathcal{B}},
\end{align}
where 
\begin{align}
    \bar{\rho}^c_{SU\Gamma Y_\mathcal{B}}\triangleq\sum\limits_{s}\sum\limits_{u}\sum\limits_{\gamma}\sum\limits_{x}\bar{P}_{SU\Gamma X|C}(s,u,\gamma,x|c)|s\rangle\langle s|\otimes|u\rangle\langle u|\otimes |\gamma\rangle\langle \gamma|\otimes\rho_{Y_\mathcal{B}}^{x},
\end{align}
with 
\begin{align}
    \bar{P}_{SU\Gamma X|C}(s,u,\gamma,x|c) 
    \triangleq P^U_S(s)P_{U|C}(u|c)P^U_\Gamma(\gamma)P_{X|f_\Gamma(X)\Gamma}(x|u,\gamma).
\end{align}

We note that by the construction of the encoders,
\begin{align}
    \bar{P}_{SU\Gamma X}(s,u,\gamma,x) & =|\mathcal{C}|^{-1}\sum_c \bar{P}_{SU\Gamma X|C}(s,u,\gamma,x|c)\\
    & =P^U_S(s)P^U_{U}(u)P^U_\Gamma(\gamma)P_{X|f_\Gamma(X)\Gamma}(x|u,\gamma).
\end{align}

We also define the classical-quantum state induced by the family of hash functions
\begin{align}
&\rho_{f_\Gamma(X)\Gamma Y_\mathcal{B}} \triangleq \sum\limits_{u}\sum\limits_{\gamma} \sum\limits_{x} P_{f_\Gamma(X)\Gamma X}(u,\gamma,x)|u\rangle\langle u|\otimes|\gamma\rangle\langle \gamma|\otimes\rho_{Y_\mathcal{B}}^{x}.
\end{align}
For the maximally mixed states $\rho_S$ and $\rho_U$, we have
\begin{align}
    \| \bar{\rho}_{SY_\mathcal{B}} -\rho_S\otimes\rho_{Y_\mathcal{B}} \|_1 &\overset{(a)}{\leq} \| \bar{\rho}_{SU\Gamma Y_\mathcal{B}} -\rho_S\otimes\rho_{U}\otimes\rho_{\Gamma Y_\mathcal{B}} \|_1\\
    & \overset{(b)}{=} \| \bar{\rho}_{U\Gamma Y_\mathcal{B}} -\rho_U\otimes\rho_{\Gamma Y_\mathcal{B}} \|_1\\
    & \overset{(c)}{\leq} \|\bar{\rho}_{U\Gamma Y_\mathcal{B}} -\rho_{f_\Gamma(X)\Gamma Y_\mathcal{B}} \|_1 + \|\rho_{f_\Gamma(X)\Gamma Y_\mathcal{B}} - \rho_U\otimes\rho_{\Gamma Y_\mathcal{B}} \|_1\\
    & \overset{(d)}{\leq} \mathbb{V}\big(\bar{P}_{U\Gamma X},P_{f_\Gamma(X)\Gamma X}\big) + \|\rho_{f_\Gamma(X)\Gamma Y_\mathcal{B}} - \rho_U\otimes\rho_{\Gamma Y_\mathcal{B}} \|_1\\
    & \overset{(e)}{=} \mathbb{V}\big(P_U^UP^U_\Gamma,P_{f_\Gamma(X)\Gamma}\big) + \|\rho_{f_\Gamma(X)\Gamma Y_\mathcal{B}} - \rho_U\otimes\rho_{\Gamma Y_\mathcal{B}} \|_1\\
    & \overset{(f)}{\leq} 2\|\rho_{f_\Gamma(X)\Gamma Y_\mathcal{B}} - \rho_U\otimes\rho_{\Gamma Y_\mathcal{B}} \|_1\\
    & \overset{(g)}{<} 4\epsilon' + 2\sqrt{2^{\log|\mathcal{U}|-H^{\epsilon'}_\text{min}(U|Y_\mathcal{B})_\psi}},\label{eq:security_constraint_intermediate}
\end{align}
where $(a)$ holds by the monotonicity of the trace distance (Lemma~\ref{lem:monotonicity_of_trace_distance} in Appendix \ref{app_I}), $(b)$ follows from Lemma~\ref{lem:quantum_state_dependence_on_classical_states} in Appendix \ref{app_H}, $(c)$ holds by the triangle inequality, $(d)$ follows from the strong convexity of the trace distance \cite[(9.72)]{wilde-2011-book}, $(e)$ holds by definition of $\bar{P}_{U\Gamma}$ and properties of variational distance (Lemma~\ref{lem:vaiational_distance_same_channel} in Appendix \ref{app_I}), $(f)$ holds since $\mathbb{V}\big(P_U^UP^U_\Gamma,P_{f_\Gamma(X)\Gamma}\big)\leq \|\rho_{f_\Gamma(X)\Gamma} - \rho_U\otimes\rho_{\Gamma} \|_1\leq \|\rho_{f_\Gamma(X)\Gamma Y_\mathcal{B}} - \rho_U\otimes\rho_{\Gamma Y_\mathcal{B}} \|_1$ by the monotonicity of the trace distance, and $(g)$ holds by Lemma~\ref{lem:lhl}.  

Then, a randomly selected encoder will be secure against the users comprising unauthorized set $\mathcal{B}$ when 
\begin{align}
    \log|\mathcal{U}| < H^{\epsilon'}_\text{min}(X|Y_\mathcal{B})_\psi. \label{eq:security_sufficient_condition_single_B}
\end{align}

\subsubsection{Achievable Rate}
We now find the achievable secret sharing rate based on the conditions derived above. We first consider the conditions \eqref{eq:shaper_uniformity_condition} and \eqref{eq:security_sufficient_condition_single_B} on the input $U$ to the shaper. By \cite[Lemma 3.2.7]{renner-2008-dissertation}, we have
\begin{align}
    H^{\epsilon'}_\text{min}(X|Y_\mathcal{B})_\psi \leq H^{\epsilon'}_\text{min}(X)_\psi,\label{eq:secrecy_redundant_with_uniform_u_condition}
\end{align} 
for all $\mathcal{B}\in\mathbb{B}$. Thus, \eqref{eq:shaper_uniformity_condition} is redundant if \eqref{eq:secrecy_redundant_with_uniform_u_condition} is satisfied for any $\mathcal{B}\in\mathbb{B}$, allowing us to jointly represent all conditions on the shaper input alphabet size  by 
\begin{align}
    \log|\mathcal{U}| < \min_{\mathcal{B}\in\mathbb{B}}H^{\epsilon'}_\text{min}(X|Y_\mathcal{B})_\psi.\label{eq:RU_constraint_ineq}
\end{align}
We select
\begin{align}
    \log|\mathcal{U}| = \min_{\mathcal{B}\in\mathbb{B}}H^{\epsilon'}_\text{min}(X|Y_\mathcal{B})_\psi - \delta-2, \label{eq:u_alphabet_selection}
\end{align}
for some $\delta>0$.

Then, using \eqref{eq:alphabet_constraints}, we have
\begin{align}
    |\mathcal{S}| & = |\mathcal{U}|/|\mathcal{C}|\\
    R &\overset{(a)}{=} \log|\mathcal{U}| - \log|\mathcal{C}| \\
    &\overset{(b)}{\geq}\min_{\mathcal{B\in\mathbb{B}}}H^{\epsilon'}_\text{min}(X|Y_\mathcal{B})_\psi - \max_{\mathcal{A}\in\mathbb{A}} H_\text{max}^{\epsilon_1}\big(U|Y_\mathcal{A}\big)_\psi - \delta -2\log\frac{1}{\epsilon_2} - 2\log{(|\mathbb{A}|+1)} - 6\\
    & \overset{(c)}{\geq} \min_{\mathcal{B\in\mathbb{B}}}H^{\epsilon'}_\text{min}(X|Y_\mathcal{B})_\psi - \max_{\mathcal{A}\in\mathbb{A}} H_\text{max}^{\epsilon_1}\big(X|Y_\mathcal{A}\big)_\psi - \delta -2\log\frac{1}{\epsilon_2} - 2\log{(|\mathbb{A}|+1)} - 6,
\end{align}
where $(a)$ holds by taking the log of both sides, $(b)$ holds by our selection of $\log|\mathcal{U}|$ in \eqref{eq:u_alphabet_selection} and the bound on the compression alphabet size \eqref{eq:compression_alphabet_lower_limit}, and $(c)$ holds since $U$ is a function of $X$. Since the above relation holds for all channel input distributions $P_X$, we have
\begin{align} \label{eq:main_result_final_rate_constraint_no_secrecy}
    R \geq \max_{P_X}\Big[\min_{\mathcal{B\in\mathbb{B}}}H^{\epsilon'}_\text{min}(X|Y_\mathcal{B})_\psi - \max_{\mathcal{A}\in\mathbb{A}} H_\text{max}^{\epsilon_1}\big(X|Y_\mathcal{A}\big)_\psi - \delta -2\log\frac{1}{\epsilon_2} - 2\log{(|\mathbb{A}|\!+\!1)} - 6\Big].
\end{align}

\subsubsection{Encoder/Decoders Selection} \label{sec:encoder_decoder_selection}
Up to this point, we have been working with the reliability and security constraints averaged over the family of encoders/decoders. We now select an encoder and set of decoders and derive the resulting bounds on reliability and security.

For the reliability, \eqref{eq:u_alphabet_selection} applied to \eqref{eq:probability_of_error_wprime} gives 
\begin{align}
    \max_{\mathcal{A}\in\mathbb{A}}\sum\limits_{c\in\mathcal{C}}|\mathcal{C}|^{-1}\bar{P}_e^{c,\mathcal{A}} \leq 3\epsilon'+2^{-\delta/2-1}. \label{eq:averaged_relibility_bound}
\end{align}

For the security, by \eqref{eq:security_constraint_intermediate}  and  our selection of $\log|\mathcal{U}|$ in \eqref{eq:u_alphabet_selection}, we have 
\begin{align}
    \max_{\mathcal{B}\in\mathbb{B}} \| \bar{\rho}_{SY_\mathcal{B}} -\rho_S\otimes\rho_{Y_\mathcal{B}} \|_1 < 4\epsilon' + 2^{-\delta/2}. \label{eq:averaged_sec_bound}
\end{align}

Then, we have that 
\begin{align}
    & \sum\limits_{c\in\mathcal{C}}|\mathcal{C}|^{-1}\Big(\max_{\mathcal{A}\in\mathbb{A}}\bar{P}_e^{c,\mathcal{A}} + \max_{\mathcal{B}\in\mathbb{B}}\| \bar{\rho}^c_{SY_\mathcal{B}} -\rho_S\otimes\rho_{Y_\mathcal{B}} \|_1\Big)\nonumber\\ 
    &\qquad= \sum\limits_{c\in\mathcal{C}}|\mathcal{C}|^{-1}\max_{\mathcal{A}\in\mathbb{A}}\bar{P}_e^{c,\mathcal{A}} + \sum_{c\in\mathcal{C}}|\mathcal{C}|^{-1}\max_{\mathcal{B}\in\mathbb{B}}\| \bar{\rho}^c_{SY_\mathcal{B}} -\rho_S\otimes\rho_{Y_\mathcal{B}} \|_1\\
    & \qquad \leq \sum_{\mathcal{A}\in\mathbb{A}}\sum\limits_{c\in\mathcal{C}}|\mathcal{C}|^{-1}\bar{P}_e^{c,\mathcal{A}} + \sum_{\mathcal{B}\in\mathbb{B}}\sum_{c\in\mathcal{C}}|\mathcal{C}|^{-1}\| \bar{\rho}^c_{SY_\mathcal{B}} -\rho_S\otimes\rho_{Y_\mathcal{B}} \|_1\\
    &\qquad \leq |\mathbb{A}|\max_{\mathcal{A}\in\mathbb{A}}\sum\limits_{c\in\mathcal{C}}|\mathcal{C}|^{-1}\bar{P}_e^{c,\mathcal{A}} + |\mathbb{B}|\max_{\mathcal{B}\in\mathbb{B}}\sum_{c\in\mathcal{C}}|\mathcal{C}|^{-1}\| \bar{\rho}^c_{SY_\mathcal{B}} -\rho_S\otimes\rho_{Y_\mathcal{B}} \|_1\\
    & \qquad  \leq  |\mathbb{A}|(3\epsilon'+2^{-\delta/2-1}) + |\mathbb{B}|(4\epsilon' + 2^{-\delta/2}), \label{eq:averaged_bound_together}
\end{align}
where the last inequality follows from  \eqref{eq:averaged_sec_bound} and \eqref{eq:averaged_relibility_bound}. It follows from \eqref{eq:averaged_bound_together} that there exists  $c^*\in\mathcal{C}$ such that 
\begin{align}
    \max_{\mathcal{A}\in\mathbb{A}}\bar{P}_e^{c^*,\mathcal{A}} &\leq |\mathbb{A}|(3\epsilon'+2^{-\delta/2-1}) + |\mathbb{B}|(4\epsilon' + 2^{-\delta/2}), \label{eq:average_probability_of_error_combined_upperbound}\\
    \max_{\mathcal{B}\in\mathbb{B}}\| \tilde{\rho}_{SY_\mathcal{B}} -\rho_S\otimes\rho_{Y_\mathcal{B}} \|_1 &\leq |\mathbb{A}|(3\epsilon'+2^{-\delta/2-1}) + |\mathbb{B}|(4\epsilon' + 2^{-\delta/2})\label{eq:average_security_combined_upperbound},
\end{align}
where $\tilde{\rho}_{SY_\mathcal{B}}=\bar{\rho}^{c^*}_{SY_\mathcal{B}}$.

Letting $\Enc=\Enc_{c^*}$ and $\Dec_\mathcal{A}=\Dec_{c^*,\mathcal{A}}$ for each $\mathcal{A}\in\mathbb{A}$ gives the result.


\section{Proof of Theorem~\ref{thm:general_asymptotic_secret_sharing_converse}}\label{sec:proof_of_converse}

We follow the approach of \cite{renes-2011-ChannelCoding,qi-2018-PrivateQBrodcast} by proving a converse for a compound secret key distribution task, where the transmitter communicates over a noisy compound quantum channel with a legitimate user in the presence of an eavesdropper, and  each authorized set $\mathcal{A} \in \mathbb{A}$ specifies a channel to the legitimate receiver and each unauthorized set $\mathcal{B} \in \mathbb{B}$ specifies a channel to the eavesdropper. 

In this secret key distribution task, the transmitter first prepares the maximally correlated state
\begin{align}
    \bar{\Phi}_{SS'} = \frac{1}{|\mathcal{S}|} \sum_{s } |s\rangle\langle s|_S \otimes |s\rangle\langle s|_{S'},
\end{align}
then for an encoder that maps the reference system $S'$ to a channel input $X$ according to $P_{X|S'}$, and transmission over the channel indexed by $\mathcal{D} \in \mathbb{A} \cup \mathbb{B}$, the induced state  is
\begin{align}
\psi_{SY_{\mathcal{D}}}
    =
    \frac{1}{|\mathcal{S}|} \sum_{s } \sum_{x } P_{X|S'}(x|s)\,
    |s\rangle\langle s|_S  \otimes \rho^{x}_{Y_{\mathcal{D}}}.
\end{align}
After the decoding procedure, the imperfect shared randomness between the transmitter and the receiver is described by
\begin{align}
    \sigma^\mathcal{A}_{SS'} = \frac{1}{|\mathcal{S}|} \sum_{s } \sum_{s' } P_{\mathcal{A}}(s'|s)\, |s\rangle\langle s|_S \otimes |s'\rangle\langle s'|_{S'},
\end{align}
where $P_{\mathcal{A}}(s'|s)$ is the probability that the receiver decodes $s'$ given that $s$ was sent by the transmitter over the channel indexed by $\mathcal{A}$.
Then, we have
\begin{align}
\left\|\sigma^\mathcal{A}_{SS'}-\bar{\Phi}_{SS'}\right\|_1 & =\sum_s\sum_{s'}\left|P^U(s)P_\mathcal{A}(s'|s)-P^U(s)\mathbf{1}\{s=s'\}\right|\\
    & = \frac{1}{|\mathcal{S}|} \sum_s\sum_{s'}\left|P_\mathcal{A}(s'|s)-\mathbf{1}\{s=s'\}\right|\\
    & = \frac{1}{|\mathcal{S}|} \sum_s\sum_{s'\neq s}P_\mathcal{A}(s'|s) + \frac{1}{|\mathcal{S}|} \sum_s\left|P_\mathcal{A}(s|s)-1\right|\\
    & = \frac{1}{|\mathcal{S}|}\sum_s P_e^{\mathcal{A}}(s) + \frac{1}{|\mathcal{S}|} \sum_s\left(1-P_\mathcal{A}(s|s)\right)\\
    & = \frac{2}{|\mathcal{S}|}\sum_s P_e^{\mathcal{A}}(s)\\
    & \leq 2\epsilon, \label{eq:converse_reliability_trace_distance_bound}
\end{align}
where the inequality holds for an $\epsilon$-good code.

Then, we have 
\begin{align}
    H_{\max}^{\sqrt{2\epsilon}}(S|Y_{\mathcal{A}})_\psi & \overset{(a)}{\leq} H_{\max}^{\sqrt{2\epsilon}}(S|S')_{\sigma^\mathcal{A}} \overset{(b)}{\leq} H_{\max}(S|S')_{\bar{\Phi}} \overset{(c)}{=} 0, \label{eq:converse_inf_rec}
\end{align}
where $(a)$ follows from the data processing inequality, see Lemma~\ref{lem:max_entropy_data_processing} in Appendix \ref{app_I}, and $(b)$ follows from the definition of the smooth max entropy since $\bar{\Phi}_{SS'}\in \mathcal{B}_{\sqrt{2\epsilon}}(\sigma_{SS'}^\mathcal{A})$, which follows from \eqref{eq:converse_reliability_trace_distance_bound} and the bound on purification distance, see Lemma~\ref{lem:purification_trace_distance_relation} in Appendix \ref{app_I}, and $(c)$ follows from the definition of $\bar{\Phi}$.

Then, for an unauthorized set $\mathcal{B}$, the security condition of an $\epsilon$-good code gives
\begin{align}
    \left\|\psi_{SY_{\mathcal{B}}}-\psi_S\otimes\sigma_{Y_\mathcal{B}}\right\|\leq \epsilon,
\end{align}
which   implies, by Lemma~\ref{lem:purification_trace_distance_relation},  that $\psi_S\otimes\sigma_{Y_\mathcal{B}}\in \mathcal{B}_{\sqrt{2\epsilon}}(\psi_{SY_{\mathcal{B}}})$. This gives, 
\begin{align}
    H_{\min}^{\sqrt{\epsilon}}(S|Y_{\mathcal{B}})_\psi & \overset{(a)}{\geq} H_{\min}(S|Y_{\mathcal{B}})_{\psi_S\otimes\sigma_{Y_\mathcal{B}}} \overset{(b)}{=} H_{\min}(S)_\psi = \log_2|\mathcal{S}|=R, \label{eq:converse_privacy_amp}
\end{align}
where $(a)$ follows from the definition of the smooth min entropy and the fact that $\psi_S\otimes\sigma_{Y_\mathcal{B}}\in \mathcal{B}_{\sqrt{2\epsilon}}(\psi_{SY_{\mathcal{B}}})$, and $(b)$ follows since $\psi_S\otimes\sigma_{Y_\mathcal{B}}$ is a product state.

As \eqref{eq:converse_inf_rec} and \eqref{eq:converse_privacy_amp} hold for all $\mathcal{A}\in\mathbb{A}$ and $\mathcal{B}\in\mathbb{B}$, we have
\begin{align}
    R &\leq \min_{\mathcal{B}\in\mathbb{B}}H_{\min}^{\sqrt{\epsilon}}(S|Y_\mathcal{B})_\psi - \max_{\mathcal{A}\in\mathbb{A}}H_{\max}^{\sqrt{2\epsilon}}(S|Y_\mathcal{A})_\psi\\
    & \leq \max_{P_{SX}}\left[\min_{\mathcal{B}\in\mathbb{B}}H_{\min}^{\sqrt{\epsilon}}(S|Y_\mathcal{B})_\psi - \max_{\mathcal{A}\in\mathbb{A}}H_{\max}^{\sqrt{2\epsilon}}(S|Y_\mathcal{A})_\psi\right]. \label{eq:converse_rate_bound_intermediate}
\end{align}

\section{Conclusion}\label{sec:conclusion}

We proved one-shot achievable rates and converse bounds for secret sharing with general monotone access structures over a classical-quantum broadcast channel. We also presented a second order expansion of our  achievable secret sharing rate  and used it to derive an asymptotically achievable secret sharing~rate.  In the special case of a classical channel where all shares are required for reconstruction, our results recover the previously established secret-sharing capacity.

Future work includes extending this work to entanglement assisted secret sharing and to the sharing of quantum secrets using a quantum channel.

\appendices
\section{Proof of Theorem~\ref{thm:general_second_order_asymptotics}} \label{app:second_order_asymptotic_proof}

We first review the notions of hypothesis testing relative entropy and smooth conditional hypothesis testing entropy \cite{tomamichel-2013-SecondOrder}. Then, we present several lemmas including second order expansions of the smooth conditional hypothesis testing relative entropy and the smooth conditional max entropy. Finally, we use these results to derive the second-order expansion for secret sharing.

Let $\rho\in\mathcal{S}_=(\mathcal{H}), \sigma\in\mathcal{P}(\mathcal{H})$ and $0\leq\epsilon\leq1$. The $\epsilon$-hypothesis testing relative entropy \cite{tomamichel-2013-SecondOrder} is defined as  $2^{-D_h^\epsilon(\rho\|\sigma)}\triangleq\inf\{\langle Q,\sigma\rangle | 0\leq Q\leq 1 \wedge \langle Q,\rho\rangle \geq 1-\epsilon\}$, where $\langle L,M\rangle=\Tr(L^\dagger M)$ is the Hilbert-Schmidt inner product. Then, for $\rho_{AB}\in\mathcal{S}_=(\mathcal{H}_{AB}), \sigma_B\in\mathcal{S}_=(\mathcal{H}_B)$, and $0\leq\epsilon\leq1$, the conditional $\epsilon$-hypothesis testing entropies \cite{tomamichel-2013-SecondOrder} are defined as $H_h^\epsilon(A|B)_{\rho|\sigma}\triangleq-D_h^\epsilon(\rho_{AB}\|\mathds{1}_A\otimes\sigma_B)$ and $H_h^\epsilon(A|B)_\rho\triangleq \max_\sigma H_h^\epsilon(A|B)_{\rho|\sigma}$.   Additionally, $2^{H_h^\epsilon(A|B)_\rho}$ is equal to  \cite[Section~II-C]{tomamichel-2013-SecondOrder}  
\begin{align}
      \label{eq:primal_hypothesis_testing_entropy_optimization}\min_{\substack{0 \le Q_{AB} \le 1 \\ \operatorname{Tr}[Q_{AB} \rho_{AB}] \ge 1-\epsilon}} \|Q_B\|.
\end{align}

We will first present a slightly modified version of Theorem~\ref{thm:general_one_shot_secret_sharing_ach} by characterizing the compression rate in terms of the smooth conditional hypothesis testing entropy in place of the smooth conditional max entropy. 

\begin{lemma}\label{lem:compound_source_coding_hypothesis_testing_entropy}
        Let $\epsilon>0$ be given. For a sequence of classical-quantum states $\{\psi_{XB^1},\dots,\psi_{XB^k}\}$, there exists an $\epsilon$-reliable source code with compound quantum side information satisfying 
    \begin{equation}
        \log|\mathcal{C}|=\max\limits_{i\in[1:k]}  \left\lceil H_h^{\epsilon_1-\eta}(X|B^i)_{\psi_{XB^i}|\psi_{XB^i}}+\log\frac{\epsilon_1}{\eta^2}\right\rceil + 2,
    \end{equation}
    where $\epsilon_1\geq \eta>0$ and $\epsilon_1=\epsilon/(k+1)$.
\end{lemma}
\begin{IEEEproof}
    The proof draws heavily from the proof of \cite[Theorem~9]{tomamichel-2013-SecondOrder} and \cite{hayashi-2003-CapacityCQChannels}, and is a modified version of Lemma~\ref{lem:one-shot-source-code-compound-quantum-side-information}. We present the parts here that differ from the proof of Lemma~\ref{lem:one-shot-source-code-compound-quantum-side-information}. We will use the following lemma.
    \begin{lemma}[\cite{hayashi-2003-CapacityCQChannels}] \label{lem:upper_bound_on_povm}
        For any $d>0$, $0\leq S\leq1$ and $T\geq 0$, we have $1-(S+T)^{-\frac{1}{2}}S(S+T)^{-\frac{1}{2}}\leq (1+d)(1-S) + (2+d+\frac{1}{d})T$.
    \end{lemma}

  For the quantum system $B^i$ for $i\in[1:k]$, with $\psi_{XB^i}$, $\mathcal{F}$, as in the proof of Lemma~\ref{lem:one-shot-source-code-compound-quantum-side-information}, we choose $Q_{XB^i}=\sum_{x\in\mathcal{X}}|x\rangle\langle x|_X\otimes Q^x_{B^i}$, where $Q_{XB^i}$ is the solution of \eqref{eq:primal_hypothesis_testing_entropy_optimization} for $H_h^{\epsilon_1-\eta}(X|B^i)_{\psi_{XB^i}|\psi_{XB^i}}$. 
    We use the decoding POVMs $\{\Lambda^{x;c,f}_{B^i}\}_{x;c,f}$, where
    \begin{align}
        \Lambda^{x;c,f}_{B^i} = \mathbf{1}\{f(x)=c\}\left(\sum_{x':f(x')=c}Q_{B^i}^{x'}\right)^{-\frac{1}{2}}Q_{B^i}^x\left(\sum_{x':f(x')=c}Q_{B^i}^{x'}\right)^{-\frac{1}{2}}.
    \end{align}
    By Lemma~\ref{lem:upper_bound_on_povm}, we have
    \begin{align}
        \mathds{1}_{B^i} - \Lambda^{x;c,f}_{B^i} \leq (1+d)\left(\mathds{1}_B-Q_{B^i}^x\right) + \left(2+d+\frac{1}{d}\right)\sum_{x'\neq x}\mathbf{1}\{c=f(x')\}Q_{B^i}^{x'},
    \end{align}
    for some $d>0$ that we will optimize over later.

    This allows us to bound the average probability of error for the state $\psi_{XB^i}$ as 
    \begin{align}
        \bar{p}_e^{B^i} &\triangleq \frac{1}{|\mathcal{F}|}\sum\limits_{f,x}P(x)\Tr\Big[\big(\mathds{1}_{B^i}-\Lambda^{x;f(x),f}_{B^i}\big)\varphi^x_{B^i}\Big]\\
        &\leq \frac{1+d}{|\mathcal{F}|}\sum_{x,f}P(x)\Tr\left[\left(\mathds{1}_{B^i}-Q_{B^i}^x\right)\varphi_{B^i}^x\right]\nonumber\\
        &\qquad\qquad+ \frac{2+d+\frac{1}{d}}{|\mathcal{F}|}\sum_{x,f}\sum_{x'\neq x}P(x)\mathbf{1}\{f(x)=f(x')\}\Tr\left[\varphi_{B^i}^xQ_{B^i}^{x'}\right]\\
        & \overset{(a)}{\leq} (1+d)\Tr\left[\left(\mathds{1}_{XB^i}-Q_{XB^i}\right)\psi_{XB^i}\right] \!+\! \left(2+d+\frac{1}{d}\right)2^{-m}\sum_{x}P(x)\Tr\left[\varphi_{B^i}^x\sum_{x'\neq x}Q_{B^i}^{x'}\right]\\
        & \leq (1+d)\Tr\left[\left(\mathds{1}_{XB^i}-Q_{XB^i}\right)\psi_{XB^i}\right] + \left(2+d+\frac{1}{d}\right)2^{-m}\sum_{x}P(x)\Tr\left[\varphi_{B^i}^xQ_{B^i}^{x}\right]\\
        & \overset{(b)}{\leq} (1+d)(\epsilon_1-\eta) + \left(2+d+\frac{1}{d}\right)2^{-m}\Tr\left[\varphi_{B^i}Q_{B^i}\right]\\
        & \overset{(c)}{=} (1+d)(\epsilon_1-\eta) + \left(2+d+\frac{1}{d}\right)2^{-m}2^{H_h^{\epsilon_1-\eta}(X|B^i)_{\psi_{XB^i}|\psi_{XB^i}}},
    \end{align}
    where $(a)$ follows by the 2-universal property of the hash family, $(b)$ follows from the because $\Tr\left[Q_{XB^i}\psi_{XB^i}\right]\geq 1-(\epsilon_1-\eta)$ by our choice of $Q_{XB^i}$, and $(c)$ follows from the choice of $Q_{XB^i}$.

    Hence, the average probability of error is bounded by $\epsilon_1$ if we choose  
    \begin{align}
        \log|\mathcal{C}|=m=\left\lceil H_h^{\epsilon_1-\eta}(X|B^i)_{\psi_{XB^i}|\psi_{XB^i}}+\log\frac{\epsilon_1}{\eta^2} + 2 \right\rceil,
    \end{align}
    where we have chosen $d=\frac{\eta}{2\epsilon_1-\eta}$. 

    The argument to find a single encoder that satisfies the reliability condition on average for each possible side information realization is the same as in the proof of Lemma~\ref{lem:one-shot-source-code-compound-quantum-side-information}, giving the result.
\end{IEEEproof}

We can now obtain the following achievable one-shot secret sharing rate for arbitrary monotone access structures.
\begin{lemma} \label{lem:achievability_hypothesis_testing_bound}
    For secret sharing over a classical-quantum broadcast channel $W$ among $L$ users with monotone access structure $\mathbb{A}$, there exists an $\epsilon$-good $2^R$-code satisfying
    \begin{align}
        R \geq \max_{P_X}\Big[\min_{\mathcal{B\in\mathbb{B}}}H^{\epsilon'}_\textnormal{min}(X|Y_\mathcal{B})_\psi - \max_{\mathcal{A}\in\mathbb{A}} H_h^{\epsilon_1-\eta}(X|Y_\mathcal{A})_{\psi|\psi} - \delta -\log\frac{\epsilon_1}{\eta^2} - 5\Big],
    \end{align}
    where $\epsilon_1\geq\eta>0$, $\delta>0$, $\epsilon'=\epsilon_1(|\mathbb{A}|+1)$, and $\epsilon = 4\left(|\mathbb{A}|(3\epsilon'+2^{-\delta/2-1}) + |\mathbb{B}|(4\epsilon' + 2^{-\delta/2})\right)$, for  $\psi_{XY_{[1:L]}}$ defined in Theorem~\ref{thm:general_one_shot_secret_sharing_ach}.
\end{lemma}
\begin{IEEEproof}
    The proof follows as in Section~\ref{sec:proof_of_main_achievability_result} with the following modifications.   For the source coding, Lemma~\ref{lem:compound_source_coding_hypothesis_testing_entropy} is used in place of Lemma~\ref{lem:one-shot-source-code-compound-quantum-side-information}. Propagating the corresponding changes to the bound on $\log|\mathcal{C}|$ and the change to $\epsilon'$ gives the result.
\end{IEEEproof}

We now present  second order expansions for the smooth conditional and hypothesis testing entropies. 

\begin{lemma}[Section VI-B, \cite{tomamichel-2013-SecondOrder}] 
\label{lem:second_order_expansion_of_conditional_hypothesis_testing_entropy}
    For any classical-quantum state $\rho_{AB}$ and any $0<\epsilon<1$, we have the following asymptotic characterization for large $n$ and $\eta = n^{-1/2}$:
    \begin{align}
        H_h^{\epsilon-\eta}(A^n|B^n)_{\rho^{\otimes n}|\rho^{\otimes n}} = nH(A|B)_\rho - \sqrt{nV(A|B)_\rho}\Phi^{-1}(\epsilon) + \mathcal{O}\left(\log n\right).
    \end{align}
\end{lemma}

\begin{lemma}[Section VI-B, \cite{tomamichel-2013-SecondOrder}] \label{lem:second_order_expansion_of_max_relative_entropy}
    For any classical-quantum state $\rho_{AB}$ and any $0<\epsilon<1$, we have the following asymptotic characterization for large $n$:
    \begin{align}
        H_{\min}^\epsilon(A^n|B^n)_{\rho^{\otimes n}} = nH(A|B)_\rho + \sqrt{nV(A|B)_\rho}\Phi^{-1}(\epsilon^2) + \mathcal{O}(\log n).
    \end{align}
\end{lemma}

We are now ready to prove the result.
\begin{IEEEproof}[Proof of Theorem~\ref{thm:general_second_order_asymptotics}]
The classical-quantum state describing the system is $\psi_{XY_{[1:L]}}^{\otimes n}$. We  bound the communication rate as
\begin{align}
    R & \overset{(a)}{\geq} \frac{1}{n}\max_{P_X} \Big[\min_{\mathcal{B\in\mathbb{B}}}H^{\epsilon'}_\textnormal{min}(X^n|Y_\mathcal{B}^n)_{\psi^{\otimes n}} - \max_{\mathcal{A}\in\mathbb{A}} H_h^{\epsilon_1-\eta}(X^n|Y_\mathcal{A}^n)_{\psi^{\otimes n}} \!-\!\delta \!+\!\log\frac{\epsilon_1}{\eta^2} \!-\! 5\Big]\\
    & \overset{(b)}{=} \frac{1}{n}\max_{P_X} \Bigg[\min_{\mathcal{B\in\mathbb{B}}}\left(nH(X|Y_\mathcal{B})_{\psi} + \sqrt{nV(X|Y_\mathcal{B})_{\psi}}\Phi^{-1}({\epsilon'}^2)\right) + \mathcal{O}(\log n)-\delta \nonumber\\
    & \qquad\qquad\qquad- \max_{\mathcal{A}\in\mathbb{A}} \left(nH(X|Y_\mathcal{A})_{\psi} -\sqrt{nV(X|Y_\mathcal{A})_{\psi}}\Phi^{-1}(\epsilon_1)\right) \!+\!\log\frac{\epsilon'_1}{\eta^2} \!-\! 5\Bigg]\\
    & \overset{(c)}{=} \max_{P_X} \Bigg[\min_{\mathcal{B\in\mathbb{B}}}\left(H(X|Y_\mathcal{B})_{\psi} + \sqrt{\frac{1}{n}V(X|Y_\mathcal{B})_{\psi}}\Phi^{-1}({\epsilon'}^2)\right)  \nonumber\\
    & \qquad\qquad\qquad- \max_{\mathcal{A}\in\mathbb{A}} \left(H(X|Y_\mathcal{A})_{\psi} -\sqrt{\frac{1}{n}V(X|Y_\mathcal{A})_{\psi}}\Phi^{-1}(\epsilon_1)\right) + \mathcal{O}\left(\frac{\log n}{n}\right)\Bigg]
\end{align}
where $(a)$ follows from Lemma~\ref{lem:achievability_hypothesis_testing_bound}, $(b)$ follows from Lemmas~\ref{lem:second_order_expansion_of_max_relative_entropy} and \ref{lem:second_order_expansion_of_conditional_hypothesis_testing_entropy}, and $(c)$ follows from the choice of $\delta=\log n$.
\end{IEEEproof}

\section{Proof of Corollary~\ref{cor:asymptotic_main_result}} \label{app:proof_asymptotic_general_secret_sharing_achievability_proof}
From Theorem~\ref{thm:general_second_order_asymptotics}, we have 
\begin{align}
    R & \geq \max_{P_X} \Bigg[\min_{\mathcal{B\in\mathbb{B}}}\left(H(X|Y_\mathcal{B})_{\psi} + \sqrt{\frac{1}{n}V(X|Y_\mathcal{B})_{\psi}}\Phi^{-1}({\epsilon'}^2)\right)  \nonumber\\
    & \qquad\qquad\qquad- \max_{\mathcal{A}\in\mathbb{A}} \left(H(X|Y_\mathcal{A})_{\psi} -\sqrt{\frac{1}{n}V(X|Y_\mathcal{A})_{\psi}}\Phi^{-1}(\epsilon_1)\right) + \mathcal{O}\left(\frac{\log n}{n}\right)\Bigg].
\end{align}

Choosing $\epsilon_1=1/\sqrt{n}$ and $\delta=\log n$ and taking the limit as $n\rightarrow\infty$ gives the asymptotic rate 
\begin{align}
    R \geq \max_{P_X} \left[\min_{\mathcal{B\in\mathbb{B}}}H(X|Y_\mathcal{B})_{\psi} - \max_{\mathcal{A}\in\mathbb{A}} H(X|Y_\mathcal{A})_{\psi}\right]. \label{eq:asymptotic_rate}
\end{align}

The asymptotic probability of error $\epsilon$ is
\begin{align}
    \lim_{n\rightarrow\infty}\epsilon & = \lim_{n\rightarrow\infty} |\mathbb{A}|(3\epsilon'+2^{-\delta/2-1}) + |\mathbb{B}|(4\epsilon' + 2^{-\delta/2})\\
    & = 0,
\end{align}
where the limit follows from the choices of $\epsilon_1$ above and the choice of $\delta=\log n$ from the proof of Theorem~\ref{thm:general_second_order_asymptotics}.

As $\epsilon$ is an upper bound on both the security and reliability, the rate in \eqref{eq:asymptotic_rate} is achievable.

\section{Proof of Corollary~\ref{cor:nonasymptotic_all_user_access_structure}}
\label{app:proof_nonasymptotic_all_user_access_structure_achievability_proof}
The proof follows the same structure as that of Theorem~\ref{thm:general_one_shot_secret_sharing_ach} but for the source coding we apply Lemma~\ref{lem:classical_data_compression_with_quantum_side_information} in Appendix \ref{app_I} in place of Lemma~\ref{lem:one-shot-source-code-compound-quantum-side-information} as there is only one authorized set, and it yields a smaller additive terms.

Application of Lemma~\ref{lem:classical_data_compression_with_quantum_side_information} replaces \eqref{eq:compression_alphabet_lower_limit} with
\begin{align}
    \log|\mathcal{C}|= \big\lceil H_\textnormal{max}^{\epsilon_1}(X|Y_{[1:L]})_\psi-2\log\epsilon_2\big\rceil + 3,
\end{align}
for $\epsilon_1,\epsilon_2\geq 0$ and allows us to replace $\epsilon'=\epsilon_1/(|\mathbb{A}|+1)+\epsilon_2$ with $\epsilon'=\epsilon_1+\epsilon_2$. When applied in the derivation of the achievable rate,  \eqref{eq:main_result_final_rate_constraint_no_secrecy} gives
\begin{align}
    R \geq \max_{P_X}\Big[\min_{\mathcal{B\in\mathbb{B}}}H^{\epsilon'}_\text{min}(X|Y_\mathcal{B})_\psi -  H_\text{max}^{\epsilon_1}\big(X|Y_{[1:L]}\big)_\psi - \delta -2\log\frac{1}{\epsilon_2} - 6\Big].
\end{align}

Lemma~\ref{lem:information_reconcilliation_to_channel_coding} is used in the same way, but we note that there is a single element in authorized set allowing to reduce \eqref{eq:averaged_bound_together} to 
\begin{align}
    \sum\limits_{c\in\mathcal{C}}\frac{1}{|\mathcal{C}|}\Big(\max_{\mathcal{A}\in\mathbb{A}}P_e^{c,\mathcal{A}} + \max_{\mathcal{B}\in\mathbb{B}}\| \tilde{\rho}_{S\Gamma Y_\mathcal{B}} -\rho_S\otimes\rho_{\Gamma Y_\mathcal{B}} \|_1\Big) \leq 3\epsilon'+2^{-\delta/2-1} + |\mathbb{B}|(4\epsilon' + 2^{-\delta/2}),
\end{align}
showing that there is an encoder and decoder for this protocol that is $\epsilon$-good for $\epsilon=3\epsilon'+2^{-\delta/2-1} + |\mathbb{B}|(4\epsilon' + 2^{-\delta/2})$.

\section{Proof of Corollary~\ref{cor:classical_secret_sharing_capacity_all_user}}
\label{app:proof_classical_capacity} 

The achievability follows  from Corollary~\ref{cor:asymptotic_all_user_access_structure} because the von Neumann entropy is equal to the Shannon entropy for classical random variables.

The converse follows from the converse for the degraded compound wiretap channel \cite{liang-2009-CompoundWTC},  by observing that our channel is equivalent to a degraded wiretap channel, where each realization corresponds to a set $\mathcal{B}\in\mathbb{B}$,  the legitimate receiver observes $Y_{[1:L]}$, the eavesdropper observes $Y_{\mathcal{B}}$, the channel is  degraded  because $X-Y_{[1:L]}-Y_\mathcal{B}$ forms a Markov chain for all $\mathcal{B}\in\mathbb{B}$, and  our definition of security \eqref{eq:def_asypmtotic_security} implies security for the degraded wiretap channel. 
Indeed, for all $\mathcal{B}\in\mathbb{B}$, we have
\begin{align}
    \frac{1}{n}I(S;Y_\mathcal{B}^n)  &\leq\frac{1}{n}\mathbb{V}(P_{SY_\mathcal{B}^n},P_SP_{Y_\mathcal{B}^n})\ln\frac{2^{nR}}{\mathbb{V}(P_{SY_\mathcal{B}^n},P_SP_{Y_\mathcal{B}}^n)}\\
    & =\mathbb{V}(P_{SY_\mathcal{B}^n},P_SP_{Y_\mathcal{B}^n})\ln 2^R - \frac{1}{n}\mathbb{V}(P_{SY_\mathcal{B}^n},P_SP_{Y_\mathcal{B}^n})\ln\mathbb{V}(P_{SY_\mathcal{B}^n},P_SP_{Y_\mathcal{B}^n})\\
    & \xrightarrow{n\to \infty} 0,
\end{align}
 where the inequality holds by \cite[Lemma 1]{csiszar-1996-AlmostIndependence}. 

\section{Proof of Lemma~\ref{lem:one-shot-source-code-compound-quantum-side-information}} \label{app:proof_of_lem_source_coding_with_compound_quantum_side_information}

Note that each decoder is a quantum measurement of the quantum system $B^i$ conditioned on a value of $C$, taking the form of a positive operator valued measure, a collection of positive operators $\Lambda^{x;c}_{B^i}$ satisfying $\sum_x\Lambda^{x;c}_{B^i}=\mathds{1}_{B^i}$ for each $c\in\mathcal{C}$. Thus, we have
\begin{align}
    p_e= 1-\min_{i\in[1:k]}\bigg\{\sum\limits_{x\in\mathcal{X}}P(x)\Tr\big[\Lambda^{x;g(x)}_{B^i}\varphi^x_{B^i}\big]\bigg\}.
\end{align}

We first state supporting lemmas.

\begin{lemma}[Lemma 1, \cite{renes-2012-OneShotCompression}]\label{lem:average_error_probability_bound}
    Let $X\in\mathcal{X}$ be a classical random variable and $B$ be a quantum system described by $\psi_{XB}=\sum_{x\in\mathcal{X}}P(x)|x\rangle\langle x|\otimes\varphi^x_B$   and $\varphi_B=\sum_{x\in\mathcal{X}}P(x)\varphi^x_B$. Let $\mathcal{F}$ be a 2-universal  family of hash functions $f:\mathcal{X}\rightarrow\{0,1\}^m$, and $P_{XB}=\sum_{x\in\mathcal{X}}|x\rangle\langle x|\otimes\Pi^x_B$ with $0\leq\Pi^x_B\leq\mathds{1}_B$ for all $x\in\mathcal{X}$. Then, there exists a family of measurements on $B$ indexed by $f\in\mathcal{F}$ and $c\in\{0,1\}^m$, having elements $\Lambda^{x;c,f}_B$ corresponding to outcomes $x$, such that $\Lambda^{x';c,f}_B=0$ when $f(x')\neq c$ and for which $\Bar{p}_e$, the probability of error averaged over a random choice of $f\in\mathcal{F}$, obeys
    \begin{align}
        \bar{p}_e&\triangleq\frac{1}{|\mathcal{F}|}\sum\limits_{f,x}P(x)\Tr\Big[\big(\mathds{1}-\Lambda^{x;f(x),f}_B\big)\varphi^x_B\Big]\\
        & \leq 2\Tr\big[(\mathds{1}-P_{XB})\psi_{XB}\big]+4\cdot2^{-m}\Tr\big[P_{XB}(\mathds{1}_X\otimes\varphi_B)\big].
    \end{align}
\end{lemma}

The next result is a corollary of results proved in the context of hypothesis testing in \cite{audenaert-2007-HypothesisTesting,audenaert-2008-HypothesisTesting}. Let $\{A\}_+$ and $\{A\}_-$ denote the projections onto the support of the positive and nonpositive parts of $A$, respectively.
\begin{lemma}[Lemma 2, \cite{renes-2012-OneShotCompression}] \label{lem:hypothesis_testing_trace_bound}
    For $\rho,\sigma\geq0$ and any $0\leq s\leq1$
    \begin{align}
        \Tr\big[\rho\{\rho-\sigma\}_-+\sigma\{\rho-\sigma\}_+\big]\leq\Tr\big[\rho^s\sigma^{1-s}\big].
    \end{align}
\end{lemma}

The following result will enable us to bound the compression alphabet size by the smooth max entropy.
\begin{lemma}[Lemma 3, \cite{renes-2012-OneShotCompression}]\label{lem:max_entropy_smoothing}
    Let $\epsilon>0$ and let $\rho_{XB}$ be a classical-quantum state. Then, there exists a classical-quantum state $\Bar{\rho}_{XB}\in \mathcal{B}_\epsilon(\rho_{XB})$ such that $H_{\max}^\epsilon(X|B)_\rho=H_{\max}(X|B)_{\Bar{\rho}}$.
\end{lemma}

We are now ready to proceed with the proof of Lemma~\ref{lem:one-shot-source-code-compound-quantum-side-information}.
\begin{IEEEproof}[Proof of Lemma~\ref{lem:one-shot-source-code-compound-quantum-side-information}]  
  Fix a family of 2-universal hash functions $\mathcal{F}$ consisting of hash functions $f:\mathcal{X}\rightarrow\{0,1\}^m$ with $|\mathcal{F}|\geq k+1$.

    Consider the quantum system $B^i$ for $i\in[1:k]$. Given $\psi_{XB^i}$, $\mathcal{F}$, and $P_{XB^i}=\sum_{x\in\mathcal{X}}|x\rangle\langle x|\otimes\Pi^x_{B^i}$ for some projectors $\Pi^x_{B^i}$ satisfying $0\leq\Pi^x_{B^i}\leq\mathds{1}^{B^i}$ for all $x\in\mathcal{X}$,  Lemma~\ref{lem:average_error_probability_bound} gives a family of measurements $\Lambda^{x;c,f}_{B^i}$ satisfying
    \begin{align}
        \Bar{p}_e^{B^i}&\triangleq\frac{1}{|\mathcal{F}|}\sum\limits_{f}p_e^{B^i}(f)\\
        & \triangleq\frac{1}{|\mathcal{F}|}\sum\limits_{f,x}P(x)\Tr\Big[\big(\mathds{1}-\Lambda^{x;f(x),f}_{B^i}\big)\varphi^x_{B^i}\Big]\\
        & \leq 2\Tr\big[(\mathds{1}-P_{XB^i})\psi_{XB^i}\big]+4\cdot2^{-m}\Tr\big[P_{XB^i}(\mathds{1}_X\otimes\varphi_{B^i})\big]\\
        & \overset{(a)}{\leq}\sqrt{8\cdot2^{-m}}\Tr\Big[\sqrt{\psi_{XB^i}}\sqrt{\mathds{1}_X\otimes\varphi_{B^i}}\Big]\\
        & \overset{(b)}{\leq} \sqrt{8\cdot2^{-m}}\bigg\|\sqrt{\psi_{XB^i}}\sqrt{\mathds{1}_X\otimes\varphi_{B^i}}\bigg\|_1\\
        & \leq \sqrt{8\cdot 2^{-m}} \max\limits_{\sigma_{B^i}}F\big(\psi_{XB^i}, \mathds{1}_X\otimes\sigma_{B^i}\big)\\
        & =\sqrt{8\cdot 2^{-(m-H_\text{max}(X|B^i)_\psi)}},
    \end{align}
    where $(a)$ follows from Lemma~\ref{lem:hypothesis_testing_trace_bound} with  $P_{XB^i}=\big\{\psi_{XB^i}-2^{-(m-1)}\mathds{1}_X\otimes\varphi_{B^i}\big\}_+$, and $(b)$ follows from $\max_U\big|\Tr[UA]\big|=\|A\|_1$ where $U$ is unitary \cite{wilde-2011-book}. The following bound on $m$ guarantees $\Bar{p}^{B^i}_e\leq \epsilon/(k+1)$:
    \begin{align}
        m & \geq H_\text{max}(X|B^i)_\psi-2\log\epsilon + 2\log{(k+1)} + 3 \label{eq:r_single_condition}.
    \end{align}

    Consider the construction of a protocol for a state $\Bar{\psi}=\sum_x\bar{P}(x)|x\rangle\langle x|\otimes\bar{\varphi}^x_{B^i}\in \mathcal{B}_{\epsilon_1}(\psi)$ for some $\epsilon_1>0$ and assume that it has average error probability bounded by $\epsilon_2/(k+1)>0$. Then, the probability of error for the state ${\psi}$ using this protocol satisfies 
    \begin{align}
        &\frac{1}{|\mathcal{F}|}\sum\limits_{f,x}P(x)\Tr\Big[\big(\mathds{1}-\Lambda^{x;f(x),f}_{B^i}\big){\varphi}^x_{B^i}\Big] \nonumber\\
        & \qquad \overset{(a)}{=} \frac{1}{|\mathcal{F}|}\sum\limits_{f}\Tr\Big[\big(\mathds{1}-\sum\limits_x|x\rangle\langle x|\otimes\Lambda^{x;f(x),f}_{B^i}\big)\psi\Big]\\
        & \qquad \overset{(b)}{\leq} \frac{1}{|\mathcal{F}|}\sum\limits_{f}\Tr\Big[\big(\mathds{1}-\sum\limits_x|x\rangle\langle x|\otimes\Lambda^{x;f(x),f}_{B^i}\big)\bar\psi\Big] + \|\bar\psi-\psi\|_1\\
        & \qquad \overset{(c)}{=} \frac{1}{|\mathcal{F}|}\sum\limits_{f,x}\bar P(x)\Tr\Big[\big(\mathds{1}-\Lambda^{x;f(x),f}_{B^i}\big)\bar{\varphi}^x_{B^i}\Big] + \|\bar\psi-\psi\|_1\\
        & \qquad \leq \epsilon_2/(k+1) + \epsilon_1,
    \end{align}
    where $(a)$ and $(c)$ follows from Lemma~\ref{lem:cq_state_equivalence} in Appendix \ref{app_G} and $(b)$ follows from  Lemma~\ref{lem:measurement_on_close_states} in Appendix \ref{app_I}. By Lemma~\ref{lem:max_entropy_smoothing}, there exists a classical-quantum state $\Bar{\psi} \in \mathcal{B}_{\epsilon_1}( \psi)$ such that $H_{\max}^{\epsilon_1}(X|B^i)_\psi=H_{\max}(X|B^i)_{\Bar{\psi}}$. Then, by setting $\epsilon/(k+1)=\epsilon_1+\epsilon_2/(k+1)$, we have $\Bar{p}_e^{B^i}<\epsilon/(k+1)$ when $m \geq H_\text{max}^{\epsilon_1}(X|B^i)_\psi)-2\log\epsilon_2 +2\log{(k+1)} + 3$. We note that the argument up to this point is the same as in the proof of \cite[Th. 1]{renes-2012-OneShotCompression}, but is presented here for clarity.

    In order to satisfy $\Bar{p}_e^{B^i}<\epsilon/(k+1)$ for all $i\in[1:k]$, we need the conditions on $m$ \eqref{eq:r_single_condition} for each $i\in[1:k]$ to be satisfied, which can be expressed as 
    \begin{align}
        m \geq \max\limits_{i\in[1:k]}H_\text{max}^{\epsilon_1}(X|B^i)_\psi-2\log\epsilon_2 + 2\log{(k+1)} + 3.
    \end{align}

   Hence, we choose
    \begin{align}
        m=\max\limits_{i\in[1:k]} \Big\lceil H_\text{max}^{\epsilon_1}(X|B^i)_\psi-2\log\epsilon_2 + 2\log{(k+1)}\Big\rceil + 3.\label{eq:r_condition}
    \end{align}
    
    For each $i\in[1:k]$, since $\Bar{p}_e^{B^i}<\epsilon/(k+1)$, there exist at most $|\mathcal{F}|/(k+1)$ hash functions $f'\in\mathcal{F}$ with
    \begin{align}
        p_e^{B^i}(f')>\epsilon.
    \end{align}
    Thus, there is a set of hash functions $\mathcal{F}^*$ with $|\mathcal{F}^*|\geq\frac{|\mathcal{F}|}{k+1}$ where each $f^*\in\mathcal{F}^*$ satisfies
    \begin{align}
        p_e^{B^i}(f^*)<\epsilon,
    \end{align}    
    for each $i\in[1:k]$, and thus $p_e\leq\epsilon$. Defining $g$ as an element of $\mathcal{F}^*$ and $h_i$ as the measurement associated with $\big\{\Lambda^{x;c,g}_{B^i}\big\}_{x\in\mathcal{X},c\in\{0,1\}^m}$ for each $i\in[1:k]$ yields the result.
\end{IEEEproof}

\section{Proof of Lemma~\ref{lem:information_reconcilliation_to_channel_coding}}\label{app:proof_of_information_reconcilliation_to_channel_coding}
We first get a bound on the probability of error for the source code when a uniformly distributed input is used.
\begin{align}
    &\max_{\mathcal{A}\in\mathbb{A}} \frac{1}{|\mathcal{U}|}\sum_{u\in\mathcal{U}} \Pr\Big[h_\mathcal{A}\big(\rho_{Y_\mathcal{A}}^{u}, g(u)\big)\neq u\Big] \nonumber \\ &\leq \mathbb{V}(\tilde{P}_{U},P_{U}^U) + \max_{\mathcal{A}\in\mathbb{A}} \sum_{u\in\mathcal{U}}\tilde{P}(u) \Pr\Big[h_\mathcal{A}\big(\rho_{Y_\mathcal{A}}^{u}, g(u)\big)\neq u\Big]\\
    & \leq \epsilon' + \epsilon'', \label{eq:triangle_eq_sc_pe_uniform_inputs}
\end{align}
where the last inequality follows from the assumptions of the theorem. We now construct the encoder corresponding to each $c\in\mathcal{C}$. We first note that
\begin{align}
    &\max_{\mathcal{A}\in\mathbb{A}}\sum_{u\in\mathcal{U}} \frac{1}{|\mathcal{U}|} \Pr\Big[h_\mathcal{A}\big(\rho_{Y_\mathcal{A}}^{u}, g(u)\big)\neq u\Big] \nonumber \\        & = \max_{\mathcal{A}\in\mathbb{A}} \sum\limits_{c\in\mathcal{C}}\frac{1}{|\mathcal{C}|}\sum_{u\in g^{-1}(c)}\frac{|\mathcal{C}|}{|\mathcal{U}|}\Pr\big[h_\mathcal{A}\big(\rho_{Y_\mathcal{A}}^{u}, c\big)\!\neq\! u\big]\\
    &=\max_{\mathcal{A}\in\mathbb{A}}\sum\limits_{c\in\mathcal{C}}\frac{1}{|\mathcal{C}|}\sum_{u\in g^{-1}(c)}P_{U|C}(u|c)\Pr\big[h_\mathcal{A}\big(\rho_{Y_\mathcal{A}}^{u}, c\big)\!\neq\! u\big]\\
    & = \max_{\mathcal{A}\in\mathbb{A}}\sum\limits_{c\in\mathcal{C}}\frac{1}{|\mathcal{C}|}\mathbb{E}_{P_{U|C=c}}\Big[\Pr\big[h_\mathcal{A}\big(\rho_{Y_\mathcal{A}}^{u}, c\big)\!\neq\! u\big]\Big] ,\label{eq:average_error_from_source_coding}
\end{align}
where the second equality follows by defining $P_{U|C}$ as the uniform distribution over $g^{-1}(C)$ and by the fact that $|g^{-1}(C)|=|\mathcal{U}|/|\mathcal{C}|$ by  Lemma~\ref{lem:linear_function_preimage} in Appendix \ref{app_I}. Let $\mathcal{S} \triangleq \big[1:|\mathcal{U}|/|\mathcal{C}|\big]$, then for each $c\in\mathcal{C}$ we choose $\Enc_c$ to be a bijection between $g^{-1}(c)$ and $\mathcal{S}$ and $\Dec_{c,\mathcal{A}}: \rho_{Y_\mathcal{A}}\mapsto \Enc^{-1}_c\big(h_\mathcal{A}(\rho_{Y_\mathcal{A}}, c)\big)$. Combining these encoders and decoders with \eqref{eq:average_error_from_source_coding} gives
\begin{align}
    \max_{\mathcal{A}\in\mathbb{A}}\sum_{u\in\mathcal{U}} \frac{1}{|\mathcal{U}|}\Pr\Big[h_\mathcal{A}\big(\rho_{Y_\mathcal{A}}^{u}, g(u)\big)\neq u\Big] &= \max_{\mathcal{A}\in\mathbb{A}}\sum\limits_{c\in\mathcal{C}}\frac{1}{|\mathcal{C}|}\mathbb{E}\Big[\Pr\big[\Dec_{c,\mathcal{A}}\big(\rho_{Y_\mathcal{A}}^{\Enc_c(S)}\big)\!\neq\! S\big]\Big] \\
    & = \max_{\mathcal{A}\in\mathbb{A}}\sum\limits_{c\in\mathcal{C}}\frac{1}{|\mathcal{C}|}\bar{P}_e^{c,\mathcal{A}}. \label{eq:source_coding_probability_of_error_equal_to_average_error_over_all_encoders}
\end{align}
Combining \eqref{eq:source_coding_probability_of_error_equal_to_average_error_over_all_encoders} and \eqref{eq:triangle_eq_sc_pe_uniform_inputs} gives the desired inequality.

\section{} \label{app_G}
\begin{lemma} \label{lem:cq_state_equivalence}
    Given the classical-quantum state $\psi=\sum_x P(x)|x\rangle\langle x|\otimes\varphi_Y^x$, where $\varphi^x_Y\in \mathcal{S}_=(\mathcal{H})$, and the collection of positive operators $\{\Lambda^x\}_{x\in\mathcal{X}}$ satisfying $\sum_x\Lambda^x=\mathds{1}_X$, we have 
    \begin{align}
        \textnormal{Tr}\Big[\big(\mathds{1}-\sum\limits_{x}|x\rangle\langle x|\otimes\Lambda^x\big)\psi\Big] = \sum\limits_{x}P(x)\textnormal{Tr}\Big[\big(\mathds{1}-\Lambda^x\big){\varphi}^x_{Y}\Big].
    \end{align}
\end{lemma}
\begin{IEEEproof}
We have    \begin{align}
        &\textnormal{Tr}\Big[\big(\mathds{1}-\sum\limits_{x}|x\rangle\langle x|\otimes\Lambda^x\big)\psi\Big]\\
        & = \textnormal{Tr}\Big[\big(\mathds{1}-\sum\limits_{x}|x\rangle\langle x|\otimes\Lambda^x\big)\big(\sum_{x'}P(x')|x'\rangle\langle x'|\otimes\varphi^{x'}_Y\big)\Big]\\
        & =\textnormal{Tr}\Big[\sum_{x'}P(x')|x'\rangle\langle x'|\otimes\varphi^{x'}_Y-\sum\limits_{x}P(x)|x\rangle\langle x|\otimes\Lambda^x\varphi^{x}_Y\Big]\\
        & = 1-\sum\limits_x P(x)\textnormal{Tr}\big[\Lambda^x\varphi_Y^x\big]\\
        & = \sum\limits_{x}P(x)\textnormal{Tr}\Big[\big(\mathds{1}-\Lambda^x\big){\varphi}^x_{Y}\Big].
    \end{align}
\end{IEEEproof}

\section{} \label{app_H}
\begin{lemma} \label{lem:quantum_state_dependence_on_classical_states}
    Let $\rho_{ABC}\in S(\mathcal{H}_A\otimes\mathcal{H}_B\otimes\mathcal{H}_C)$,  where $\rho_{ABC}\triangleq\sum_a\sum_bP(a,b)|a\rangle\langle a|\otimes|b\rangle\langle b|\otimes\rho_C^b$ for $\rho_C^b\in S(\mathcal{H}_C)$ and $\rho_{AB}\triangleq \sum_a\sum_bP(a,b)|a\rangle\langle a|\otimes|b\rangle\langle b|$. Then, we have 
    \begin{align}
        \|\rho_{ABC}-\rho_{AB}\otimes\rho_C\|_1 = \|\rho_{BC}-\rho_B\otimes\rho_C\|_1.
    \end{align}
\end{lemma}
\begin{IEEEproof}
Utilizing the definitions above, we have 
    \begin{align}
        &\|\rho_{ABC}-\rho_{AB}\otimes\rho_C\|_1 \nonumber\\
        & = \left\|\sum_a\sum_bP(a,b)|a\rangle\langle a|\otimes|b\rangle\langle b|\otimes\rho_C^b- \sum_a\sum_bP(a,b)|a\rangle\langle a|\otimes|b\rangle\langle b|\otimes\rho_C\right\|_1\\
        & = \left\|\sum_a\sum_bP(a,b)|a\rangle\langle a|\otimes|b\rangle\langle b|\otimes\big(\rho_C^b- \rho_C\big)\right\|_1 \label{eq:sum_in_trace_distance}\\
        & \overset{(a)}{=} \sum_a\sum_bP(a,b)\left\|\rho_C^b-\rho_C\right\|_1\\
        & = \sum_b P(b) \left\|\rho_C^b-\rho_C\right\|_1\\
        & \overset{(b)}{=} \left\|\sum_bP(b) |b\rangle\langle b|\otimes(\rho_C^b-\rho_C)\right\|_1\\
        & = \|\rho_{BC}-\rho_{B}\otimes\rho_C\|_1,
    \end{align}
    where $(a)$ and $(b)$ follow from \cite[Proposition~2.7]{khatri-2020-QuantumCommunication}.
\end{IEEEproof}

\section{Supporting Results} \label{app_I}
\begin{lemma}[Theorem 18, \cite{tomamichel-2010-Duality}] \label{lem:max_entropy_data_processing}
    Let $\epsilon\geq0$, $\rho_{AB}\in S_\leq(\mathcal{H}_{AB})$ and $\mathcal{P}(\mathcal{H}_B)\rightarrow\mathcal{P}(\mathcal{H}_D)$ be a trace preserving completely positive map with $\tau\triangleq\left(\mathds{1}_A\otimes\mathcal{E}\right)(\rho_{AB})$, then 
\begin{align}
    H_{\max}^{\epsilon}(A|B)_\rho\leq H_{\min}^{\epsilon}(A|D)_\tau.
\end{align}
\end{lemma}

\begin{lemma}[Lemma 6, \cite{tomamichel-2010-Duality}] \label{lem:purification_trace_distance_relation}
    Let $\rho,\tau\in S_\leq(\mathcal{H})$. Then 
    \begin{align}
        \frac{1}{2}\|\rho-\tau\|_1\leq P(\rho,\tau)\leq \sqrt{\|\rho-\tau\|_1}.
    \end{align}
\end{lemma}

\begin{lemma}[Lemma 4, \cite{renes-2011-ChannelCoding}]\label{lem:linear_function_preimage}
    Let $f:\mathcal{X}\rightarrow\mathcal{Y}$ be a linear function.  Then $|f^{-1}(y)|=|\mathcal{X}|/|\mathcal{Y}|$ for all $y\in\mathcal{Y}$.
\end{lemma}

\begin{lemma}[Corollary 9.1.2, \cite{wilde-2011-book}] \label{lem:monotonicity_of_trace_distance}
    Let $\rho_{AB},\sigma_{AB}\in \mathcal{S}_=(\mathcal{H}_A\otimes \mathcal{H}_B)$. The trace distance is monotone with respect to the discarding of subsystems:
    \begin{align}
        \|\rho_A-\sigma_A\|_1 \leq \|\rho_{AB}-\sigma_{AB}\|_1.
    \end{align}
\end{lemma}

\begin{lemma}[Corollary 9.1.1, \cite{wilde-2011-book}] \label{lem:measurement_on_close_states}
    Given two quantum states $\rho,\sigma\in \mathcal{S}_=(\mathcal{H})$ and an operator $\Pi\in\mathcal{L}(\mathcal{H})$ satisfying $0\leq\Pi\leq \mathds{1}$. Then, we have 
    \begin{align}
        \textnormal{Tr}\big[\Pi\sigma\big]\leq \textnormal{Tr}\big[\Pi\rho\big] + \| \rho-\sigma\|_1.
    \end{align}
\end{lemma}

\begin{lemma}[Lemma 16, \cite{cuff-2009-dissertation}] \label{lem:variational_distance_marginal_smaller}
    For two probability distributions distributions $P_{AB}$ and $\bar{P}_{AB}$, we have
    \begin{align}
        \mathbb{V}(P_{AB},\bar{P}_{AB})\geq \mathbb{V}(P_A,\bar P_A).
    \end{align}
\end{lemma}

\begin{lemma}[Lemma 17, \cite{cuff-2009-dissertation}] \label{lem:vaiational_distance_same_channel}
    For two marginal distributions $P_A$ and $\bar P_A$ and the conditional distribution $P_{B|A}$, we have
    \begin{align}
        \mathbb{V}(P_A,\bar P_A)=\mathbb{V}(P_AP_{B|A},\bar P_AP_{B|A}).
    \end{align}
\end{lemma}

\begin{definition} \label{def:source_code_quantum_side_information}
    A source code with quantum side information for the  classical quantum state $\psi_{XB}=\sum_xp(x)|x\rangle\langle x|\otimes\varphi_B^x$ consists of 
    \begin{itemize}
        \item an encoder $g:\mathcal{X}\rightarrow\mathcal{C}$,
        \item a decoder $h:\mathcal{S}_=(\mathcal{H}_{B})\times\mathcal{C}\rightarrow\mathcal{X}$,
    \end{itemize}
    has probability of error defined by
    \begin{align}
        p_e=\Pr\Big[h\big(\varphi^X_{B},g(X)\big)\neq X\Big],
    \end{align}
    and is said to be $\epsilon$-good if $p_e\leq\epsilon$.
\end{definition}

\begin{lemma}[Theorem 1, \cite{renes-2012-OneShotCompression}] \label{lem:classical_data_compression_with_quantum_side_information}
    Given  $\epsilon\geq0$ and state $\psi_{XB}$ as defined in Definition~\ref{def:source_code_quantum_side_information}, there exists an $\epsilon$-good source code with quantum side information if the compression alphabet satisfies
    \begin{align}
        \log|\mathcal{C}| = \big\lceil H_\textnormal{max}^{\epsilon_1}(X|B)_\psi-2\log\epsilon_2\big\rceil + 3,
    \end{align}
    for $\epsilon_1,\epsilon_2\geq 0$ such that $\epsilon_1+\epsilon_2=\epsilon$.
\end{lemma}

\bibliographystyle{IEEEtran}
\bibliography{../references.bib}

\end{document}